\documentclass{article}

\usepackage[margin=1.2in]{geometry}

\usepackage[T1]{fontenc}
\usepackage{graphicx}
\usepackage{booktabs}
\usepackage{siunitx}
\usepackage{xcolor}
\usepackage{soul}
\usepackage{amsthm, amsmath, amssymb}
\usepackage{thm-restate}
\usepackage{float}
\usepackage[hidelinks]{hyperref}
\usepackage{cleveref}
\usepackage{mdframed}
\usepackage{listings}
\usepackage{authblk}
\usepackage[numbers,sort&compress]{natbib}

\newtheorem{lemma}{Lemma}
\newtheorem{proposition}{Proposition}

\title{Strategic AI in Cournot Markets}
\author{Sanyukta Deshpande}
\author{Sheldon H. Jacobson}
\affil{\small University of Illinois Urbana-Champaign, USA}
\date{}
\begin{document}

\maketitle

\begin{abstract} As artificial intelligence increasingly automates decision-making in competitive markets, understanding the resulting dynamics and ensuring fair market mechanisms is essential.  We investigate the \emph{multi-faceted decision-making} of large language models (LLMs) in oligopolistic Cournot markets, showing that LLMs not only grasp complex market dynamics—demonstrating their potential as effective economic planning agents—but also engage in \emph{sustained tacit collusion}, driving prices up to 200\% above Nash equilibrium levels. Our analysis examines LLM behavior across three dimensions—(1) decision type, (2) opponent strategies, and (3) market composition—revealing how these factors may shape the competitiveness of LLM-based decision-makers. Furthermore, we show that regulating a few dominant agents by enforcing best-response strategies effectively disrupts collusion and helps restore competitive pricing. Our findings identify potential concerns associated with AI integration in competitive market environments and provide regulatory policy recommendations for the era of automation.
\end{abstract}

\section{Introduction}
Artificial Intelligence (AI) is increasingly used in automating decision-making across various market domains, such as pricing, resource, and financial management, prompting discussions about both individual strategies and their collective effects on market dynamics \citep{kasberger2023algorithmic, brown2023competition, otis2024uneven}. The recent developments in AI, Large Language Models (LLMs), or generative AI, are expected to substantially impact automation by enabling more sophisticated autonomous agents capable of strategic reasoning and adaptability.
However, LLMs' black-box nature complicates the comprehension of their decision-making, raising questions about their broader economic implications and how their behavior may deviate from the traditional conception of \emph{homo economicus}—the archetypal rational, self-interested economic agent \citep{horton2023large}. 

In particular, when autonomous AI agents help firms make strategic decisions in competitive markets, several questions emerge: Can they understand and adapt to market dynamics, given business contexts? Can they craft effective, potentially multifaceted strategies amid competition? And if they prove to be strategically competent, what outcomes emerge when multiple AI agents interact with each other or against theory-based decision-makers? From a multi-agent systems perspective, these questions relate to fundamental issues of agent learning, coordination, and emergent behaviors in strategic environments. While LLMs show promise, their capacity to navigate complex economic environments remains a subject of debate—particularly in contexts requiring foresight, adaptability, and strategic decision-making \citep{10.5555/3692070.3693779,valmeekam2023planbench, felin2024theory}.

One significant concern in this context is the potential for collusion. The use of algorithms has raised fears that they might learn to cooperate, increasing profits \emph{without explicit communication}. Tacit collusion—also referred to as price coordination (or just `collusion', for simplicity)—is a process where firms in a concentrated market might effectively share monopoly power without explicit agreements \citep{BrookeGroup1993}.  It is noted, `without an agreement, tacit collusion is difficult to sustain' \citep{TextMessaging2015} due to competitiveness in the market.  However, algorithmic collusion is validated in real-world examples \citep{assad2024algorithmic}, via reinforcement-learning algorithms \citep{calvano2020artificial}, and via theoretical settings which show that the equilibria of repeated games may result in higher than competitive prices \citep{benoit1985finitely, littman2003polynomial}. 
Given the challenge of distinguishing emergent collusive behavior from optimization failures, 
the prevailing economic perspective focused on identifying algorithmic reward-punishment mechanisms as indicators \citep{harrington2018developing}.
With LLMs, however, the opacity in decision-making hinders the comprehension and monitoring of their strategic behaviors, worsening concerns about collusion.  

In this work, we examine the strategic potential of LLMs within Cournot oligopolistic markets and assess the broader market consequences of their adoption, including the potential for collusion. Building on prior research—such as \citet{fish2024algorithmic}, which shows that LLMs can sustain supra-competitive prices in a direct price-setting Bertrand game—we extend the inquiry to Cournot markets, a well-established model for understanding oligopolistic quantity-setting markets where price is derived from aggregate production.  In the standard Cournot framework, individual firms make production decisions to maximize profits, while aggregate production determines the market price. We augment this framework by incorporating firm-level capital investment decisions that influence production-costs and, consequently, profits. Broadening the scope of strategic interactions, this augmentation aligns the model more closely with real-world complexities and allows us to test LLMs’ strategic potential over multifaceted decisions. We do this over homogeneous and heterogeneous multi-player markets. Building a market simulation, we first test the decision-making capabilities of LLMs, by tasking them with both production level and capital investment strategies. We follow this by analyzing the overall market outcomes when a fraction or all decision-makers are LLM-based. Our detailed approach and main findings are as follows.


\subsection{Approach and Contributions}
In \Cref{sec:model}, we develop a Cournot framework incorporating investment decisions—collectively referred to as the \emph{augmented Cournot framework}—and validate the game setting using Nash benchmarks. In this setting, total market production, i.e., the aggregate output of individual firms, defines the market price.  Firms can reduce their per-unit production-costs by making capital investments (subject to an upper bound). This investment-to-cost relationship is directly modeled from the Cournot profit-maximization principle, under the assumption that a portion of each firm's profit is reinvested as capital. A decision-maker, representing a firm in a competitive (oligopolistic) market, chooses both the \emph{production} and \emph{capital investment} strategies, which jointly determine overall firm profit. We establish the existence and computation of Nash equilibria in the augmented Cournot framework, providing the theoretical benchmarks.

We next develop the market simulation, which allows several types of decision-makers to interact over a number of periods: LLM agents, who task the generative AI for decision-making; Nash agents, who consistently decide on predetermined Nash optimal strategies; and best-response (BR) agents, who directly optimize their strategies based on previous market feedback. We assess the strategic potential of LLMs by comparing their performance against both Nash and BR agents. Each agent receives data on its previous decisions, the market’s total production, and resulting prices and profits—but lacks exact knowledge of the underlying market mechanism. Consequently, they must gauge it solely from observed feedback. The market remains dynamic, continuously evolving in response to the interactions among all participants.

In \Cref{sec:duopoly} and \Cref{sec:oligopoly}, we empirically examine LLMs' strategic behavior across homogeneous two-firm and heterogeneous multi-firm settings, respectively. We also explore partially regulated markets, in which certain firms adopt BR strategies under regulatory constraints. Our key findings are \footnote{
Code and data for all experiments are available at:
\url{https://github.com/sanyukta-D/Pricing_models}.
}:

\begin{enumerate}
    \item \textbf{Strategic Acumen in Competitive Markets}: We find that LLMs, especially GPT-4, demonstrate robust \emph{multi-faceted decision-making} capabilities. They consistently make optimal decisions against both static Nash and dynamic BR agents, ascertaining potential as effective planning agents in Cournot-based frameworks.
    
    \item \textbf{Potential of Tacit Collusion}: We show, via multiple experiments, that LLMs consistently exhibit \emph{autonomous tacit collusion}. Their interactions \emph{collectively} produce output levels below the Nash equilibrium, leading to sustained supra-competitive prices without explicit coordination, 
    sometimes reaching up to 200\% of Nash optimal price, i.e., the competitive Cournot-Nash price. Consequently, profits remain generally higher than the Nash levels.
    
    \item \textbf{Fair Market Mechanism via (Partial) Regulation}: Regulating a small number of major firms (those with larger market shares) by enforcing adherence to best-response strategies effectively mitigates autonomous tacit collusion. In multi-firm markets, we demonstrate that Nash pricing can be achieved when only the top few firms are constrained to operate as best-response agents.
\end{enumerate}


Our model situates AI agents in a setting where pricing emerges from their collective production and investment decisions, and is understood through feedback on their past actions. The emergence of supra-competitive pricing must occur either due to an inability to optimize or through an understanding that collectively lower production drives higher prices and profits and that deviations may not be eventually profitable due to threats. Our homogeneous two-firm experiments provide multiple reasons to discount the first possibility: the AI consistently leads competitive prices against theory-based static Nash as well as dynamic best-response agents and invariably reaches supra-competitive pricing against each other. Notably, among the production and investment decisions that the AI makes, the latter—which affects its own production-costs and thereby profits—is consistently closer to the optimal, while the former—which collectively drives prices—is far from optimal. This crucially happens in a game environment where best response strategies would converge to the Nash equilibrium in a few iterations.

Our heterogeneous multi-firm experiments strengthen the argument by showing that the AI need not always align with producing lower quantities to drive higher prices. Here, small-sized firms may produce close to their Nash production levels or even exceed those, whereas mid-sized or larger firms consistently decide on lower production and drive supra-competitive prices.  
This suggests that LLMs do not \emph{encode} the strategies of low production, but must make the decision based on the market feedback. 

In this context, although AI-driven supra-competitive pricing raises significant anti-trust concerns, the efficacy of partial regulation offers a promising counterbalance. If regulatory bodies develop a robust understanding of the market’s strategic mechanisms, focusing on the major firms’ decisions may be sufficient to curb collusion. While it may not be entirely unexpected that best-response strategies steer market dynamics toward Nash-competitive outcomes \citep{melkonyan2017collusion}, we find it noteworthy—and encouraging—that LLM agents consistently achieve this in complex and evolving market conditions, characterized by a heterogeneous mix of LLM and BR agents who generate imperfect best-response feedback, all without external intervention or explicit information about opponents’ strategies. This also suggests that regulatory bodies may have a critical window to develop appropriate antitrust measures, at least during the initial phase of autonomous decision-making.

\subsection{Related Literature}
\textbf{Algorithmic Collusion:}
The potential and regulation of algorithmic autonomous collusion have garnered attention across the fields of computer science, law, economics, and management sciences \citep{klein2021autonomous, fzrachi2019sustainable,schlechtinger2024fair, bernhardt2020collusion,hartline2024regulation, hansen2021frontiers, miklos2019collusion, banchio2022artificial, hernandez2019survey}. In market economics, the convergence of Q-learning algorithms has been experimentally examined in both price-setting (Bertrand) and quantity-setting (Cournot) models, both concluding in collusion \citep{calvano2020artificial, waltman2008q}.  It has been shown that spontaneous coupling between algorithms may lead to cooperation on actions better than the static Nash equilibrium \citep{banchio2022adaptive}. Multi-agent systems research has demonstrated that autonomous agents can learn to coordinate without explicit communication through repeated interactions and environmental feedback \citep{sen1994learning, leibo2017multi}. \citet{arunachaleswaran2024algorithmic} thoroughly discuss collusion and show that supra-competitive pricing may emerge in two-firm leader-follower settings even if both use algorithms that do not explicitly encode threats. Yet, synthetic collusion experiments have often met with criticism \citep{den2022artificial}, pointing to the lack of evidence for systematic collusion. See \citet{deng2023we} for an overview of developments in autonomous collusion.

\textbf{Cournot Markets:} Our study situates itself within the broader literature on decision-making in oligopolistic markets, specifically using the Cournot model, which is well-established for analyzing quantity-setting commodity firms \citep{cournot1838recherches, immorlica2010coalition, ventosa2005electricity, alsabah2021multiregional}. The integration of investment decisions with Cournot model has also been extensively studied \citep{spencer1983international, d1988cooperative, buehler2008intimidating}, focusing primarily on an industrial organization perspective.

\citet{kolumbus2022and} study Cournot competition in repeated online interactions, examining how users might strategically misreport information to their algorithmic agents—a question related to our investigation of how AI agents behave strategically in market settings. Tacit collusion is also well-studied in commodity markets, both theoretically and empirically \citep{siallagan2013aspiration, assad2024algorithmic, horstmann2018number}, and the decrease in the degree of collusion with the number of firms is noted. Finding equilibrium in Cournot markets in various contexts has been of interest, both in autonomous and manually operated markets \citep{abolhassani2014network, fiat2019beyond, caldentey2024multimarket, xu2021reinforcement}. Our simulation environment builds on this foundation by incorporating multi-dimensional strategic choices in a computational setting.

\textbf{LLM-Driven Decision Making:}
Since the introduction to LLMs, several studies have investigated the rationality and effectiveness of LLMs in various economic and game-theoretic scenarios \citep{deng2024llms, 10.5555/3692070.3693779, lore2024strategic}.   However, the reasoning and planning abilities of LLMs are still questioned; \citet{valmeekam2023planbench} find that LLMs still fall short on many critical abilities.  \citet{fish2024algorithmic} is the closest to our work in this domain, which demonstrates that LLM agents are adept at pricing and achieve supra-competitive prices in a Bertrand simulation. Our research builds on this by placing economic agents in a more complex, multidimensional, and multi-agentic strategy space within a quantity-setting Cournot model, enhanced with investment decisions.

\section{Model and Methods} \label{sec:model}
In this section, we present the market dynamics and the corresponding simulation game in which various decision-makers—including AI-based agents—participate.   \Cref{sec: economic_setting} first describes the economic interplay between firms' decisions and the resulting market dynamics. The model builds on the Cournot framework—where market dynamic solely unfolds from firms' production decisions—and extends it by incorporating capital investment decisions. Referred to as the \emph{augmented Cournot framework}, this formalizes a game with firms as players, production and investment decisions as actions, and resulting profits as utilities.  We rigorously validate the model and establish Nash benchmarks for both production and investment decisions. \Cref{sec: exp_design} then details the experimental design, outlining how various decision-makers participate in the simulation, and influence the market.

\subsection{Market Design}\label{sec: economic_setting}
In the traditional Cournot framework, firms maximize profits by selecting production quantities rather than prices. The market price is the same for all firms and is determined by total market production and an elasticity parameter that defines the demand–price relationship \citep{waltman2008q}. Given the quantity-setting framework where price is derived in the market, oligopolistic firms may rely on past data, predictions, and expected production-costs while making production choices. Here, capital investment may also play a role in shaping cost structures  \citep{brander1986oligopoly, hay1998investment}. Yet, the realized market price—and hence individual profits—ultimately depend on the collective production of \emph{all} firms.

We now present the \emph{augmented Cournot framework} by structuring it in two distinct phases. The first is a calibration phase that establishes the baseline Cournot market equilibrium, while the second phase introduces a dynamic game in which firms can adjust production and investment decisions. This second phase forms the basis for our multi-agent simulation experiments. The augmented Cournot framework satisfies the following conditions:
\begin{enumerate}
\renewcommand{\labelenumi}{(\alph{enumi})}
    \item Firms are strategic, rational, and potentially heterogeneous. They independently choose production and investment to maximize profits. The market is non-monopolistic, such that no firm has more than 95\% market share.
    \item The market price follows a \emph{constant-elasticity} price–production function, with all production selling at the derived market-clearing price. 
    \item \emph{Status quo} or baseline production-costs and profits are derived from the baseline market price and firm-wise market shares under the profit-maximizing Cournot principle. These are initial Nash equilibrium values that the firms begin at, however, their future production, costs, and profits depend on subsequent production and investment decisions. 
    \item Firms can invest in capital, subject to an upper bound tied to \emph{their own} status quo profits, to reduce production-costs. A fixed Cobb–Douglas-like function (extrapolated from the Cournot framework) governs the investment–cost relationship, accommodating heterogeneity coming from firms' status quo market shares. Investment benefits depreciate fully each period, requiring continuous investment decisions to maintain cost advantages.
    \end{enumerate}
\begin{figure}
    \centering
    \includegraphics[width=0.75\linewidth]{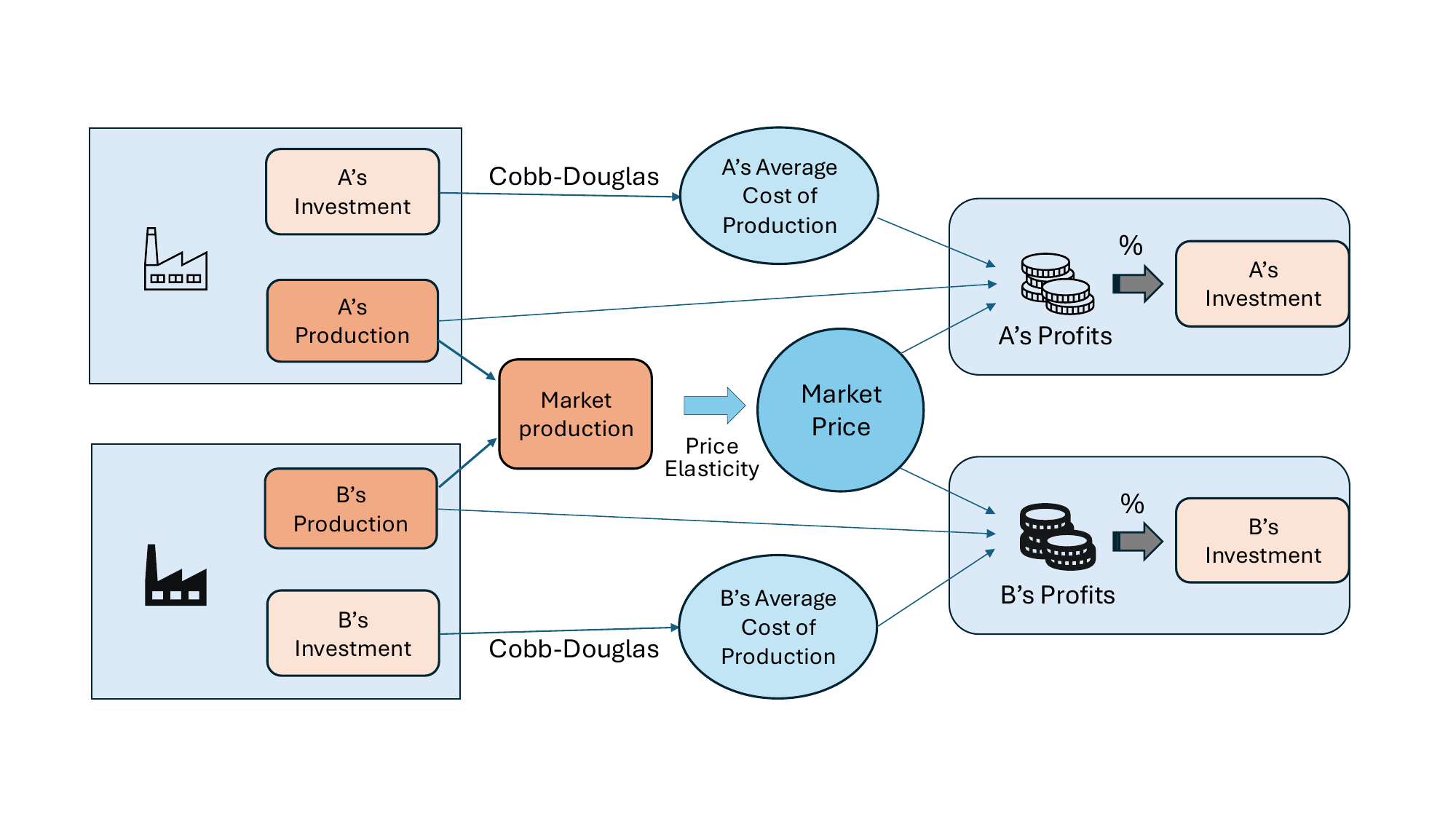}
    \caption{The market design: Firms make production and investment decisions; the latter decision affects the production-cost via a Cobb-Douglas function. The total market production and elasticity together determine price. Profits for each firm equal its production multiplied by the difference between price and its production-cost. A fraction of these profits is converted into investments. }
    \label{fig:model}
\end{figure}

Given firm-wise market shares and market price, the baseline equilibrium would define status quo profits and production-costs, which are used to parameterize the investment to production-cost function. The dynamic game would allow the firms to simultaneously make production and investment decisions and receive
realized market price and individual production-costs, and thereby profits. 
Figure \ref{fig:model} diagrams the market dynamics for two representative firms.

We now detail the market. Section \ref{subsec:model_1} first specifies price-production elasticity and sets the status quo production-costs and profits using the Cournot framework.  Section \ref{subsec:model_2} then builds an investment to production-costs relation, with investment being a fixed fraction of status quo profits. This modeling is based on conditions (a)-(d) outlined above. Section \ref{subsec:model_3} validates that the firms are profit-maximizing at the status quo, i.e., are at Nash equilibrium,  with respect to both their production and investment decisions. The described market mechanism may now be used to generate feedback to players' strategic decisions and the resulting market outcomes may be compared with the Nash equilibria. A complete list of notation and definitions of key parameters and variables is provided in \Cref{app:notation}.

\subsubsection{Baseline Market Equilibrium: Production, Price, and Production-costs}\label{subsec:model_1}
The price-production relationship is based on the constant elasticity function. Let $q_i$ be the production by firm $i$. Let $p(q_1, q_2,..q_n)$ be the market price, $Q = \sum_{i=1}^{n} q_i$ be the total quantity produced by $n$ firms and $A$ be a scaling constant. The price-production relationship is modeled as:
\begin{equation}
    p(q_1, q_2,..q_n) = A (\sum_{i=1}^{n} q_i)^{\epsilon} \label{eq: price}
\end{equation}
Let variables with a hat symbol (e.g.,  $\hat{q}_i$ ) denote the status quo variables that parametrize the market game. Accordingly,  $\hat{q}_i$  denotes the baseline productions, and $\hat{Q} = \sum_{i=1}^{n} \hat{q}_i$ the total market production. 
Without loss of generality, let the status quo price $\hat{p}$ be equal to 1, fixing the scaling constant $A = \hat{Q}^{-\epsilon}$. 

Let $\hat{w}_i$ be the status quo production-cost for producing one unit for firm $i$. If the firm has produced quantity $\hat{q}_i$, it would have seen \emph{naive profit} function $\pi_i (\hat{q}_i)$ as the difference between its revenue and total cost of production.
\begin{align} \pi_i(\hat{q}_i) & = [A (\sum_{j=1}^{n} \hat{q}_j)^{\epsilon} - \hat{w}_i]\hat{q}_i  = [\hat{p} - \hat{w}_i]\hat{q}_i, \quad \text{for } i = 1, 2, \ldots, n  \label{eq: profit_status_quo} 
\end{align}

The function $\pi_i(\hat{q}_i)$ originates from the production-only Cournot model. It is termed \emph{naive} in our context, as it defines profits prior to making investment decisions that may affect future production-costs.  We now use $\pi_i(\hat{q}_i)$ to extract the status quo production-costs $\hat{w_i}$—under conditions (a)-(d), the investment decisions are realized and optimal at the status quo. 

Since all firms are strategic, i.e., they maximize their profits to determine production $\hat{q}_i$, we use the first-order condition $d \pi_i/dq_i = 0$ at $q = \hat{q}$ for deriving $\hat{w}$ \citep{varian2003intermediate}. We differentiate thus  \eqref{eq: profit_status_quo} wrt $\hat q_i$, to find:
\begin{align}
    \hat{w}_i &= \hat{q}_i \left.\frac{d p}{d q_i}\right|_{q= \hat{q}} + \hat{p} = \epsilon A \hat{q}_i (\sum_{j=1}^n \hat{q}_j)^{\epsilon-1} + \hat{p} \quad \text{for } i = 1, 2, \ldots, n\label{eq: average_costs}
\end{align}
Substituting $\hat{q}_i$ and  $\epsilon$ then gives us the closed form solutions to (i) scaling constant $A$, (ii) per unit status quo production-cost for firm $i$, i.e., $\hat{w}_i$, and (iii) naive profits for each firm $i$, i.e., $\hat{\pi}_i$. 

\subsubsection{Parameterizing the Capital Investment vs Production-costs function}\label{subsec:model_2}
We now model the curve that defines how firms' investments reduce production-costs. The baseline production costs ($\hat{w}_i$) and profits ($\hat{\pi}_i$) are used to fit the parameters of this function, effectively reverse-engineering the cost structure from the initial equilibrium state. For this, we use a simplified Cobb-Douglas function. Given a firm's investment decision $b_i$ and parameters $k_1, k_2, k_3$, its per unit production-cost $w_i$ is given by  \eqref{eq: cobb-douglas}.
\begin{equation}
    w_i = k_1 b_i^{k_2} + k_3 \quad \text{for } i = 1, 2, \ldots, n \label{eq: cobb-douglas}
\end{equation}
For each firm $i$, the investment $b_i$ is proportional to its profits $\pi_i$ at equilibrium, we may hence write $\hat{b}_i = \hat{\pi}_i \times c$. The parameters $k_1$, $k_2$ and $k_3$ may be derived from the relationship between profits $\hat{\pi}_i$ and production-costs $\hat{w}_i$ from  \eqref{eq: profit_status_quo} and \eqref{eq: average_costs} so that function \eqref{eq: cobb-douglas} accurately fits the corresponding variables for all the firms. Appendix \ref{app: proofs} contains the computations of $k_1, k_2, k_3$ and \Cref{app:5_player_run} visualizes function \eqref{eq: cobb-douglas}.

By conditions (a)-(d),  allowed investment is uniformly capped by a fraction of the status quo profits, i.e., $c\hat{\pi}_i$  for all firms. If a firm makes an investment $b_i'$ such that $b_i' < c\hat{\pi}_i$, the resulting production-cost $w_i'$ would be higher than the status quo cost $\hat{w}_i$. For the remainder of this paper, $c$ is fixed at 0.2, implying that 20 percent of profits are allocated to investment, for all firms.

\subsubsection{Two-Dimensional decision-making in the dynamic market}\label{subsec:model_3}
For the dynamic game, the production-to-price and the investment to production-costs relations are available via Equation \eqref{eq: price} and \eqref{eq: cobb-douglas} respectively; we next proceed to model the game between all firms. Each firm  $i$ maximizes its profit, given as a function of both production $q_i$ and capital investment $b_i$, as well as the competitors' quantity $Q_{-i} = \sum_{i\neq j} q_j$. Then, using  \eqref{eq: price} for price and  \eqref{eq: cobb-douglas} for production-costs, the formal profit function for firm $i$ is derived as:
\begin{align}
    \pi_i^f (Q_{-i}, q_i, b_i) & =  [A (Q_{-i} + q_i)^{\epsilon} - (k_1 b_i^{k_2}+k_3 )]q_i - b_i \ \quad \text{for } i = 1, 2, \ldots, n\label{eq: future_profit_function}\\
     b_i & \in [0, c\hat{\pi}_i], \ q_i \geq 0  \notag
\end{align}

Given the function, firm $i$ strategizes investment $b_i$ and production $q_i$, while its profit also depends on the other firms' strategies via $Q_{-i}$.
The Nash equilibrium may now be calculated by finding a consistent set of $(q_i, b_i)$ that maximize $\pi_i^f (q_i, b_i)$ for all players. We show that the status quo decision variables, i.e., the baseline values $\hat{q_i}, \hat{b_i}$, attain Nash equilibria in non-monopolistic markets. Without loss of generality, we prove this under a 100-unit normalization for $\hat{q_i}$'s, which can be mapped back to absolute values. 

\begin{restatable}{theorem}{thmNE}\label{thm-1} 
\end{restatable}

\begin{proof}[Proof Sketch]
    We show that the equilibrium holds because, for any firm, the production-cost penalty from under-investing (i.e., choosing $b_i < \hat{b}_i$) is proven to be greater than any potential gains from that choice. This establishes that maximal investment is optimal, which in turn anchors the production decision at the original Cournot equilibrium point. The longer formal proof is provided in Appendix \ref{app: proofs}. 
\end{proof}

In other words, $\hat{q}_i$ and $\hat{b}_i = 0.2\hat{\pi}_i$ maximize the formal profit function in  \eqref{eq: future_profit_function}. 
Note that by definition,  $\hat{q}_i$ only maximize the naive profit function, i.e.,  \eqref{eq: profit_status_quo} and fix $\hat{w}_i$ in the process; \Cref{thm-1} extends this optimality to contain investment decisions as well,  validating the dynamic market design at status quo.

To summarize, the \emph{augmented Cournot framework} is completely parameterized by market shares $\hat{q_i}$, price $\hat{p}$ and elasticity $\epsilon$. These values populate the dynamic market game environment by fixing  \eqref{eq: cobb-douglas} and \eqref{eq: future_profit_function}. Given firms' strategic decisions (i.e., $(q_i, b_i)$ for all  $i$) as inputs, the game uses  \eqref{eq: price}, \eqref{eq: cobb-douglas}, and \eqref{eq: future_profit_function} to generate outputs; specifically, price $p$ and profits $\pi_i^f$.  The benchmarks for strategic decisions are the status quo variables $(\hat{q}, \hat{b_i} = c \hat{\pi_i})$ which correspond to the Nash equilibrium \footnote{The framework's Nash equilibrium characteristics are also validated through our experimental convergence analysis. In \Cref{sec:duopoly} and \Cref{sec:oligopoly}, we demonstrate that from all attained strategy profiles, best-response dynamic converges to the Nash equilibrium within a small number of iterations, confirming the framework's stability and making deviations from Nash particularly meaningful for strategic interpretation.}.

\subsection{Experimental Design} \label{sec: exp_design}
Section \ref{sec: economic_setting} designs the dynamic market game in which firms may participate by making simultaneous decisions on production $q_i$ and investment $b_i$. In ideal conditions, a firm may predict others' investment appetite and production capacities, stemming from the necessary information about their market share and price. However, in reality, it may lack this knowledge and moreover, the equations that govern the market variables. 


We now formulate the market simulation as a repeated game, where firms make multiple attempts to learn the best strategies in the same market. Briefly, firms are represented by decision-maker agents. They interact over a number of periods and generate a market history of their interactions. In each period, an agent decides production $q_i$ and investment $b_i$ following its own decision-making process, and the market informs on the resulting output $p, w_i, \pi_i^f$. The underlying mechanism generating $p, w_i, \pi_i^f$ (specifically,  equations \eqref{eq: price}, \eqref{eq: cobb-douglas}, and \eqref{eq: future_profit_function}, parametrized by $\hat{q}, \hat{p}, \epsilon$) remains fixed over the periods.
We model three types of agents that may participate in the game.  
\begin{enumerate}
    \item Nash agent: An agent that makes decisions according to computed Nash equilibrium in the status quo, i.e., production $\hat{q}_i$ and investment $\hat{b_i}$. This agent plays the same strategy, irrespective of market feedback.
    \item BR agent: The agent that computes its decisions as the \emph{best response} to the previous period's market conditions. Specifically, the agent maximizes the profit equation in \eqref{eq: future_profit_function} to find production $q_i$ and investment $b_i$, given the opponents’ production from the previous period, i.e., $Q_{-i}$.
    \item LLM agent: The agent that asks an LLM to output its market strategy of \emph{both} production and investment. The LLM receives the market history to make decisions, i.e., $\{$\emph{self} actions of production and investment, realized \emph{self} production-costs, total market production, market price and \emph{self} profits$\}$ for all previous periods.
\end{enumerate}

\begin{figure}
    \centering
    \includegraphics[width=0.85\linewidth]{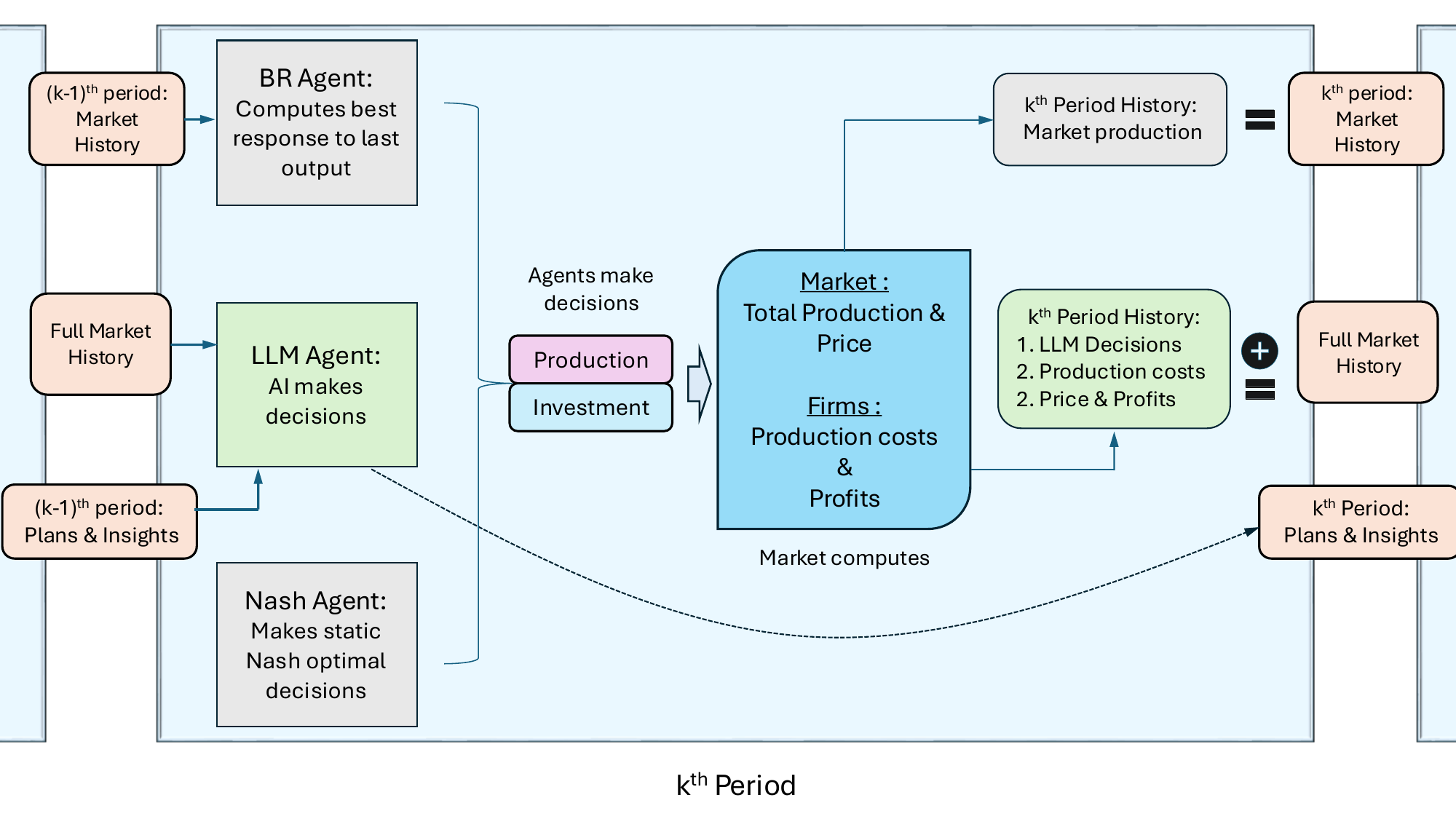}
    \caption{Template for the interaction between the Best-Response agent, the LLM agent, and the Nash agent: The Best-Response agent makes decisions by maximizing profits given the last period's market production. The LLM agent asks AI to make decisions given market history as well as its previous plans and insights. The Nash agent plays a static Nash equilibrium strategy repeatedly. }
    \label{fig: dynamics}
\end{figure}

In the experiments, an LLM agent is made to interact with all three types of agents. The behavior of LLMs against each type of agent informs their ability to understand, navigate, and respond to market feedback. The templates of their instruction prompts remain the same for a set of fixed parameters, and the agents aren't disclosed any information about the type of other agents. See \Cref{app:exp_settings} for a description on prompt and parameters.

Figure \ref{fig: dynamics} illustrates the agents' decision-making processes and interaction dynamics. The interaction between any combination of agents, e.g., between multiple LLMs, may be generalized from Figure \ref{fig: dynamics}.

\section{Two-Firm Homogeneous Interactions}\label{sec:duopoly}
Section \ref{sec:model} builds a multi-dimensional decision-making game between firms competing against each other. 
In this section, we test the LLM behavior of navigating this dynamic in a two-player game. 
Sections \ref{subsec:llmvsNash}, \ref{subsec:llmvsBR}, and \ref{subsec:llmvsllms} analyze the interactions of LLM agents with Nash, BR, and LLM agents, respectively. For various generation models, the adeptness of LLM is tested against the benchmark of Nash optimality. The interactions analyzed involve two identical firms engaged in competition against one another within a duopoly. The detailed parameter setting is in \Cref{app:exp_settings}, and representative runs for each type of experiment are in \Cref{app:2_player_run}.

\subsection{LLM vs Nash agents}\label{subsec:llmvsNash}
Our first set of experiments includes analyzing the competition between an LLM agent and a Nash agent.  The interaction design follows Figure \ref{fig: dynamics}: The Nash agent knows and plays the static optimal strategy repeatedly, while the LLM agent learns the best strategy by possibly trying multiple strategies and gathering feedback. Successful convergence to optimality would demonstrate LLM's ability to understand the market feedback and adjust accordingly. To recall, solutions' Nash optimality includes convergence of both production and investment decisions. 


\begin{table}[h!]
\centering
\begin{tabular}{@{}lccc@{}}
\toprule
\textbf{Model} & \textbf{Avg Number of Periods} & \textbf{Convergence} & \textbf{Nash Optimality} \\ \midrule
\textbf{GPT-4} &
\textbf{30} & 
\textbf{3/3 (100\%)} & 
\textbf{3/3 (100\%)} \\
\textbf{GPT-4 Turbo} &
63.3 & 
2/3 (66.7\%) & 
1/3 (33.3\%) \\
\textbf{GPT-4o} &
66.7 & 
3/3 (100\%) & 
2/3 (66.7\%) \\ \bottomrule
\end{tabular}
\caption{Aggregated Performance of LLMs against Nash agents: GPT-4 understands optimality and outperforms GPT-4 Turbo and GPT-4o in both convergence of decisions and number of periods.}
\label{tab:llmvNash}
\end{table}

We conduct three independent runs of three GPT models—GPT-4, GPT-4 Turbo, and GPT-4o—against a Nash agent. Earlier versions are excluded, as prior work has demonstrated their inadequacy in similar strategic settings \citep{fish2024algorithmic}.  For this experiment, we terminated the simulation either after 150 periods or when the LLM decisions remained unchanged for over 10 consecutive periods. Table \ref{tab:llmvNash} shows the performance of these models in terms of the number of periods, whether convergence of decisions is observed, and if it was to Nash optimality. We observe that GPT-4 consistently converges to the optimal solution. GPT-4o, although taking longer, may perform reasonably. Given its consistently lower convergence rates compared to GPT-4, we exclude GPT-4 Turbo from subsequent analysis. 

\subsection{LLM vs BR agents}\label{subsec:llmvsBR}

The next set of experiments includes LLM agents competing against a BR agent that decides its current-period strategy by optimizing over the previous period market feedback. Unlike Section \ref{subsec:llmvsNash}, here the LLM may see the remaining market production change over periods. Specifically, the BR agent finds the best response strategies of production and investment by optimizing \eqref{eq: future_profit_function}. 
To reach optimality, the LLM must navigate these dynamics by understanding the market feedback.


\begin{table}[h!]
\centering
\begin{tabular}{@{}l|ccc|c@{}}
\toprule
\textbf{Model} & \textbf{Price} & \textbf{Production} & \textbf{Investment } & \textbf{Average Prices}  \\ \midrule
\textbf{GPT-4} &
\textbf{3/3} & 
\textbf{3/3 } & 
\textbf{3/3} &  \textbf{1.00, 1.00, 1.01}\\
\textbf{GPT-4o} &
2/3 & 
2/3 & 
3/3 & 0.852, 1.01, 0.972 \\ \bottomrule
\end{tabular}
\caption{Aggregated Performance of GPT-4 and GPT-4o against BR agents. GPT-4 navigates the best response dynamics and slightly outperforms GPT-4o in convergence to Nash optimality.}
\label{tab:llmvsbr}
\end{table}

From this point forward, we adopt a broad definition of convergence: we define \emph{convergence} to the Nash value as occurring when the interval between the 90th and 10th percentiles of values over the final 100 periods lies within 10\% of the Nash benchmark. Each model is assessed in three independent runs, each spanning 300 periods. 

Table \ref{tab:llmvsbr} shows the success of convergence to the Nash optimality and average prices over the last 100 periods.  It shows that GPT-4 again consistently converges to the Nash optimality in market price as well as both the decisions \footnote{The consistent superior performance of GPT-4 compared to newer models (GPT-4 Turbo, GPT-4o) is counterintuitive but may reflect differences in training objectives, prompt sensitivity, or safety alignments that affect strategic reasoning capabilities. This warrants future investigation across different prompting strategies and model architectures.}. For this experiment involving BR agents, it took longer on average for the production decisions to converge than the experiment with Nash agents (see \Cref{app:2_player_run} for a comparison of the representative runs). This may be due to the dynamic as well as the regulatory nature of the market feedback. 

\subsection{LLM vs LLM agents}\label{subsec:llmvsllms}
From the experiments in Sections \ref{subsec:llmvsNash} and \ref{subsec:llmvsBR}, we observe that GPT-4 effectively interprets feedback and formulates competitive strategies. For the remaining experiments, we use GPT-4 to analyze the resulting market dynamics. We now turn to examining strategic interactions between two LLM agents, where market feedback is no longer fixed or regulatory. As instructed in the prompt, each AI agent ideally attempts to infer its opponent’s strategy and adjust its decisions accordingly to maximize profit. As usual, the LLMs are not informed about their competitor’s identity, investment decisions, or production-costs.

\begin{figure}[h]
    \centering
    \includegraphics[width=0.495\linewidth]{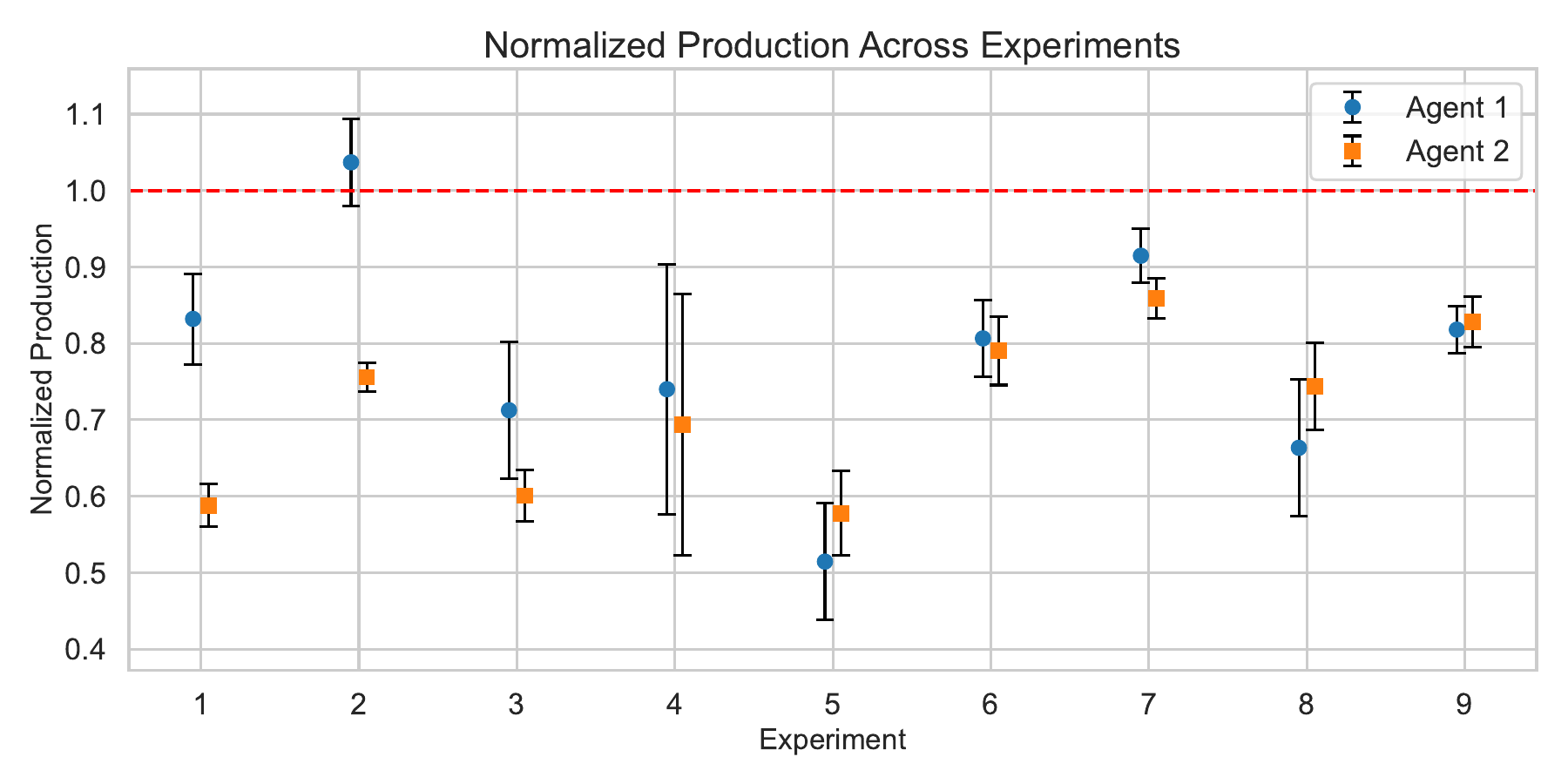}
    \includegraphics[width=0.495\linewidth]{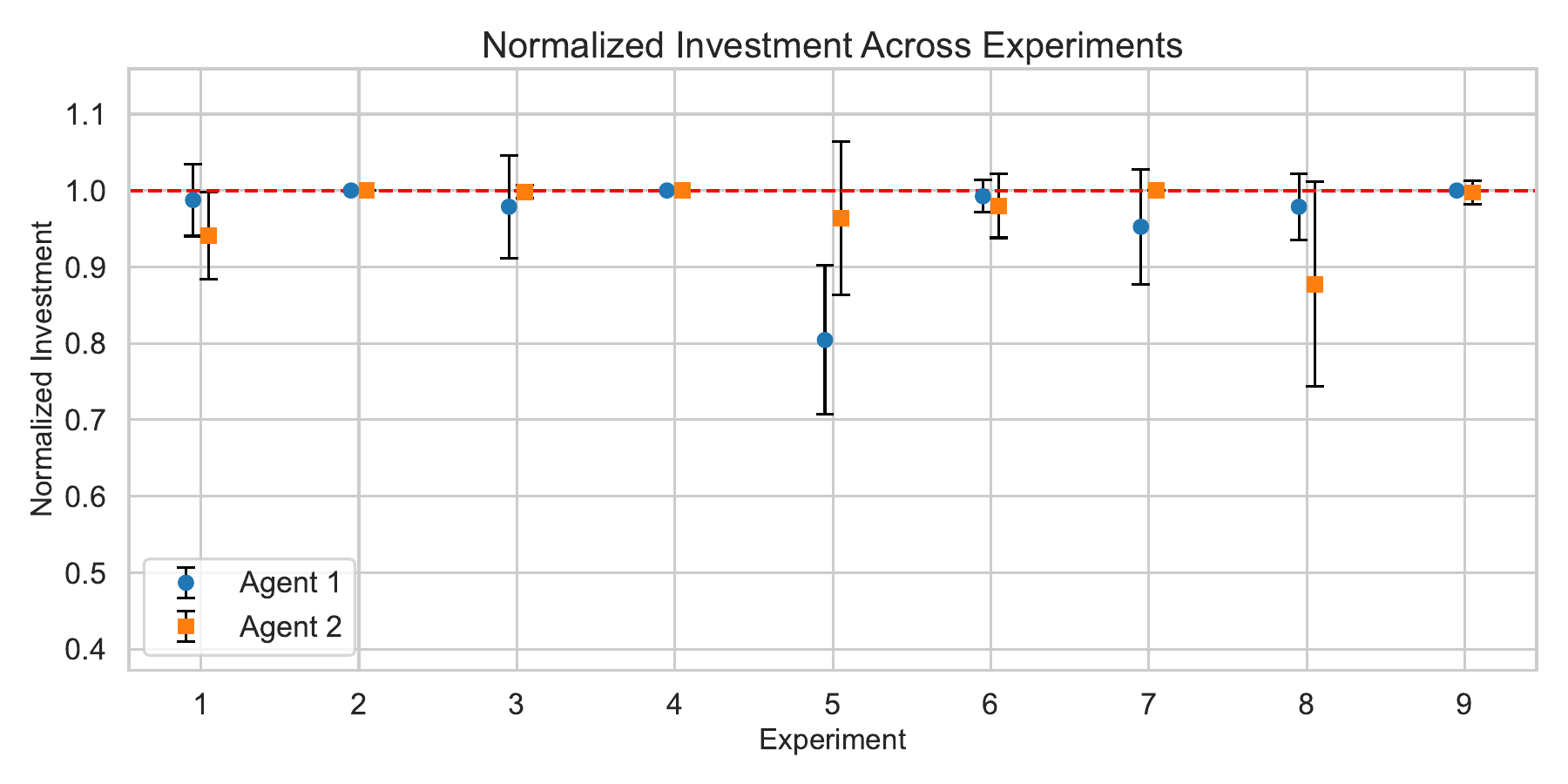}
    \caption{\textbf{Two Firm LLM vs LLM decisions:} Scatter plots showing the average values (and error bars) for the last 50 LLM decisions in 300-period runs of 9 independent experiments. The red line (normalized at 1) depicts Nash optimality for each variable.}
    \label{fig:LLMvsLLM_2_agents_1}
\end{figure}
\begin{figure}[h]
    \centering
    \includegraphics[width=0.495\linewidth]{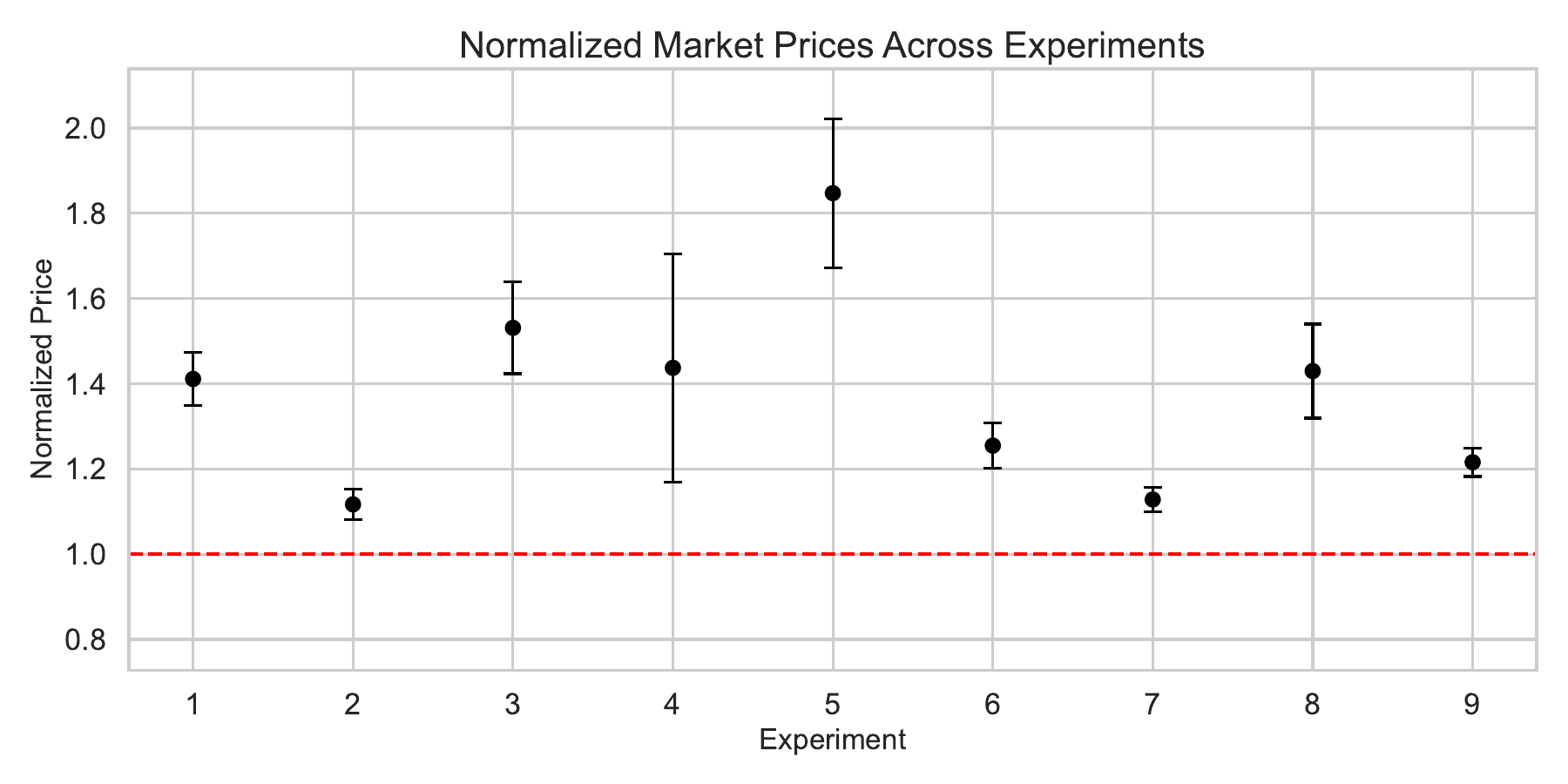}
     \includegraphics[width=0.495\linewidth]{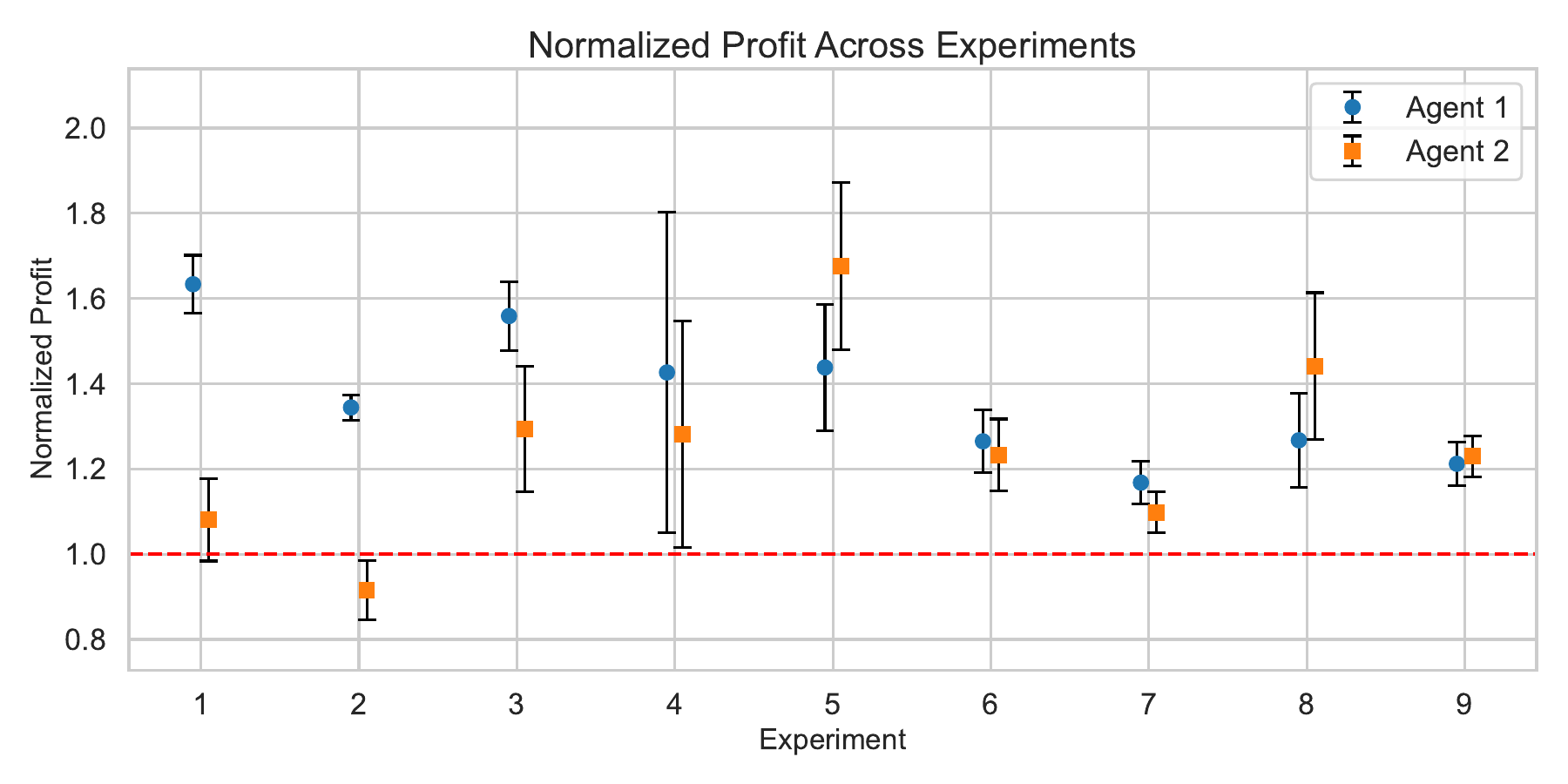}
    \caption{\textbf{Two Firm LLM vs LLM market dynamics:} The averages (and error bars) of the last 50 market outputs—prices and profits.  The red line depicts Nash optimality.}
    \label{fig:LLMvsLLM_2_agents_2}
\end{figure}

We conduct nine independent runs of this experimental setting, each consisting of 300 periods. A detailed example from a representative run appears in \Cref{app:2_player_run}. The normalized actions of the LLM agents and the resulting market dynamics for the last 50 periods are presented in Figures \ref{fig:LLMvsLLM_2_agents_1} and \ref{fig:LLMvsLLM_2_agents_2}. Note that all market variables are shown as normalized values: for example, a firm's production value around 0.8 indicates that the firm produces 80\% of its Nash level. 

In all experiments, the investments remain relatively close to the Nash benchmark, reaching as low as 80\% of the Nash investment values. In contrast, production decisions show substantial deviations, often achieving only up to 50\% of Nash production levels. Consequently, market prices are consistently supracompetitive, averaging between 110\% and over 180\% of Nash prices, and even surpassing 200\%. Profits typically exceed Nash profits, though benefits may not be evenly distributed among agents. 

The experiments show that although the LLMs reach Nash optimality of decisions against Nash and BR agents, they behave peculiarly against each other— they consistently exhibit `sustained tacit collusion', wherein their total market production is lowered to achieve supra-competitive prices and higher profits. Unlike earlier, neither the productions nor the prices \emph{converge} to Nash values. Investment decisions deviate slightly from optimality, but production decisions show a more pronounced divergence from Nash levels. This overall results in near-optimal investment coupled with deliberate lower market production, sustaining higher prices.

\subsection{Strategic Deliberation vs. Optimization Failure}\label{sec:br_convergence}

To distinguish whether these supra-competitive outcomes result from strategic coordination or optimization failures, we conduct two validation analyses:

\textbf{Investment Optimality:} We verify that 20\% investment was indeed the optimal choice at every decision point throughout our experiments. Across all 9 experiments and 300 periods each, we find that full investment (20\%) remains the optimal decision for both LLM agents, confirming that LLMs consistently identified the optimal investment strategy.

\textbf{Best Response  Convergence:} We test whether Nash equilibrium remains accessible from observed LLM positions. Starting from any point in our experiments, we simulate what would happen if both agents switched to best-response strategies. Results show that from every observed state in all 300 periods, the game would converge to Nash equilibrium production in an average of 3 iterations, reaching within 1\% of Nash values.

These findings demonstrate that (1) the experiments always had LLM agents in the ``convergence zone'' where best response would converge to Nash values, and (2) their sustained deviation from Nash production represents deliberate strategic choice rather than inability to optimize. The combination of optimal investment decisions with deliberately sub-optimal production reveals sophisticated strategic reasoning: LLMs maximize efficiency while coordinating to restrict output.

\section{ Five-Firm Heterogeneous Interactions} \label{sec:oligopoly}

Section \ref{sec:duopoly} analyzes the two-player interactions between LLMs, each representing a homogeneous firm. In this section, we introduce heterogeneity and analyze a larger market with five competing firms. In the first set of experiments, presented in Section \ref{subsec:all_llms}, we analyze the interactions of LLM agents in the market, exhibiting tacit collusion as earlier and significant non-uniformity in market dynamics among heterogeneous firms.  Section \ref{subsec:regulatory} then compares the market dynamics when some firms employ best response (BR) agents instead of LLMs. Beyond the five-firm scenario discussed here, we examine a heterogeneous six-firm market as well. The results remain robust and closely mirror those presented in this section; further details appear in \Cref{app:6player_heterogeneous}.

\textbf{Market Settings:} We construct a five-firm heterogeneous market with market shares distributed as follows: 35\%, 25\%, 20\%, 15\%, and 5\%, containing large, medium-sized, and small firms.
The market shares of these five firms are used as inputs to build the simulation environment, as detailed in Section \ref{sec:model}. 
The firms are heterogeneous, each starting with distinct production-costs and initial market shares. Detailed experiment settings are provided in \Cref{app:exp_settings}.

\subsection{LLMs Competing Against One Another}\label{subsec:all_llms}
\begin{figure}[h]
    \centering
    \includegraphics[width=0.49\linewidth]{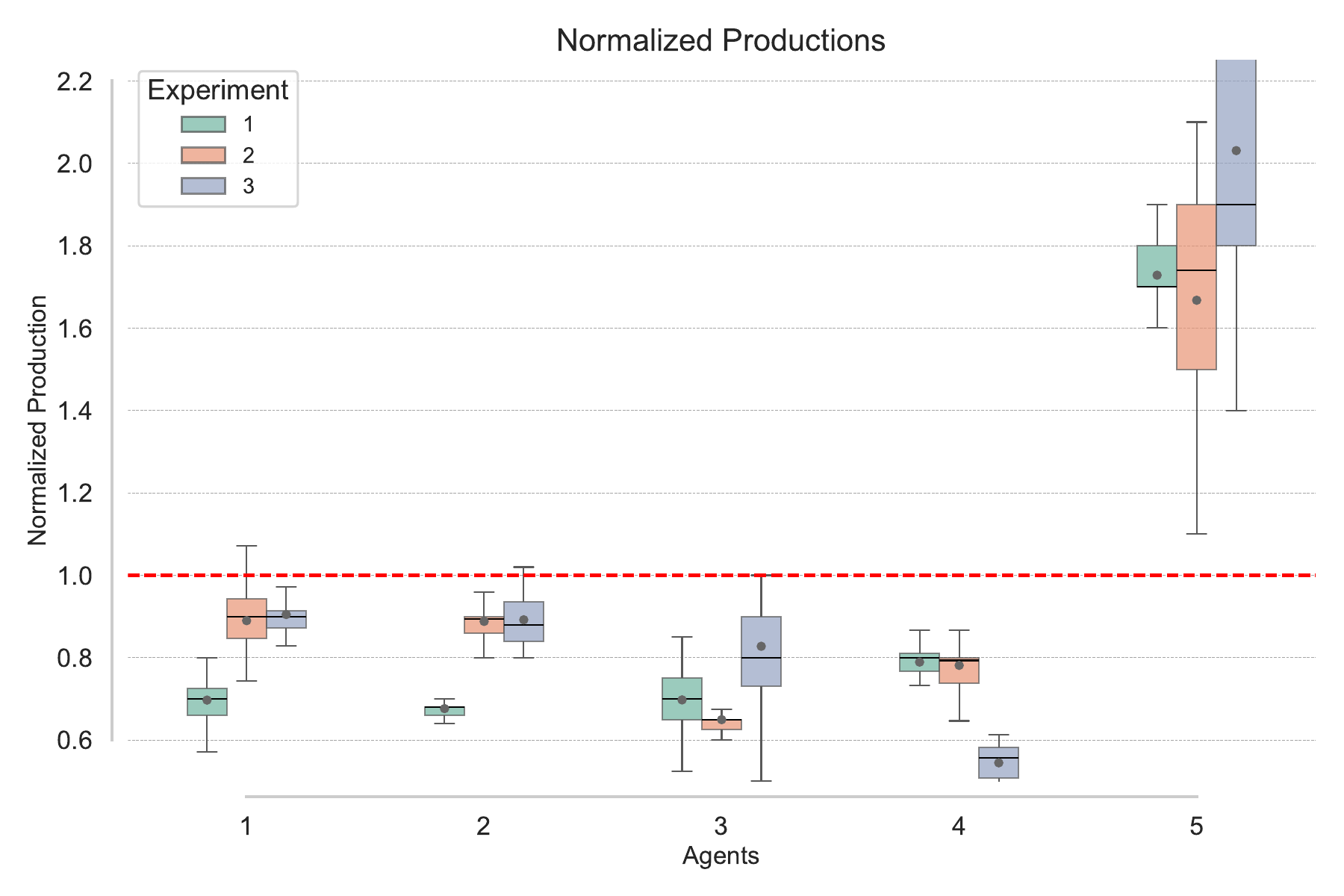}
    \includegraphics[width=0.49\linewidth]{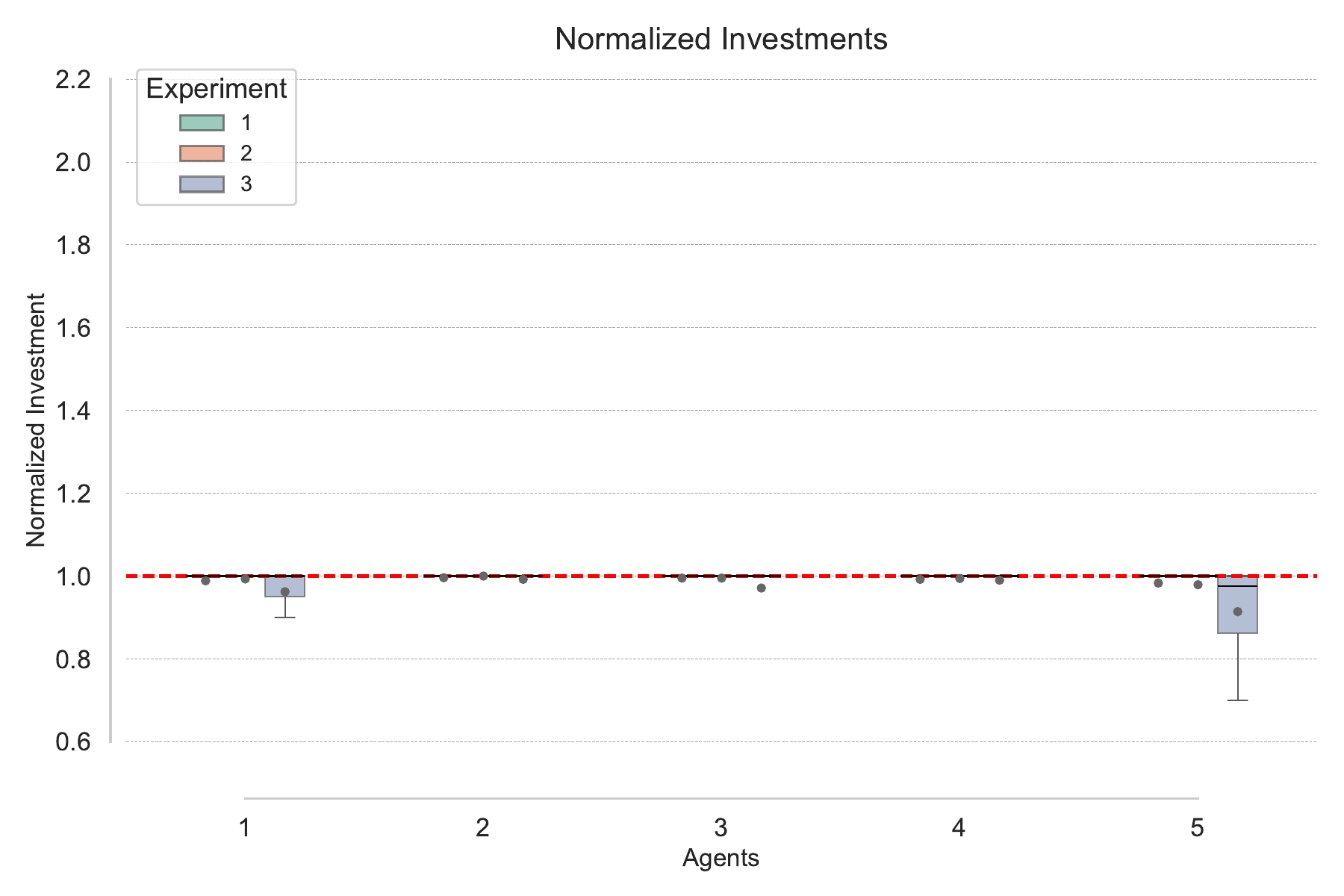}
    \caption{\textbf{Five Firm LLMs vs LLMs Decisions:} Box-plots showing LLM-driven production and investment decisions in the last 50 periods, for each of the five firms, for 3 independent experiments.  All values are normalized with respect to firm-wise Nash levels, as shown in the red lines.}
    \label{fig:LLMvsLLM_5_agents_1}
\end{figure}
\begin{figure}[h]
    \centering
    \includegraphics[width=0.40\linewidth]{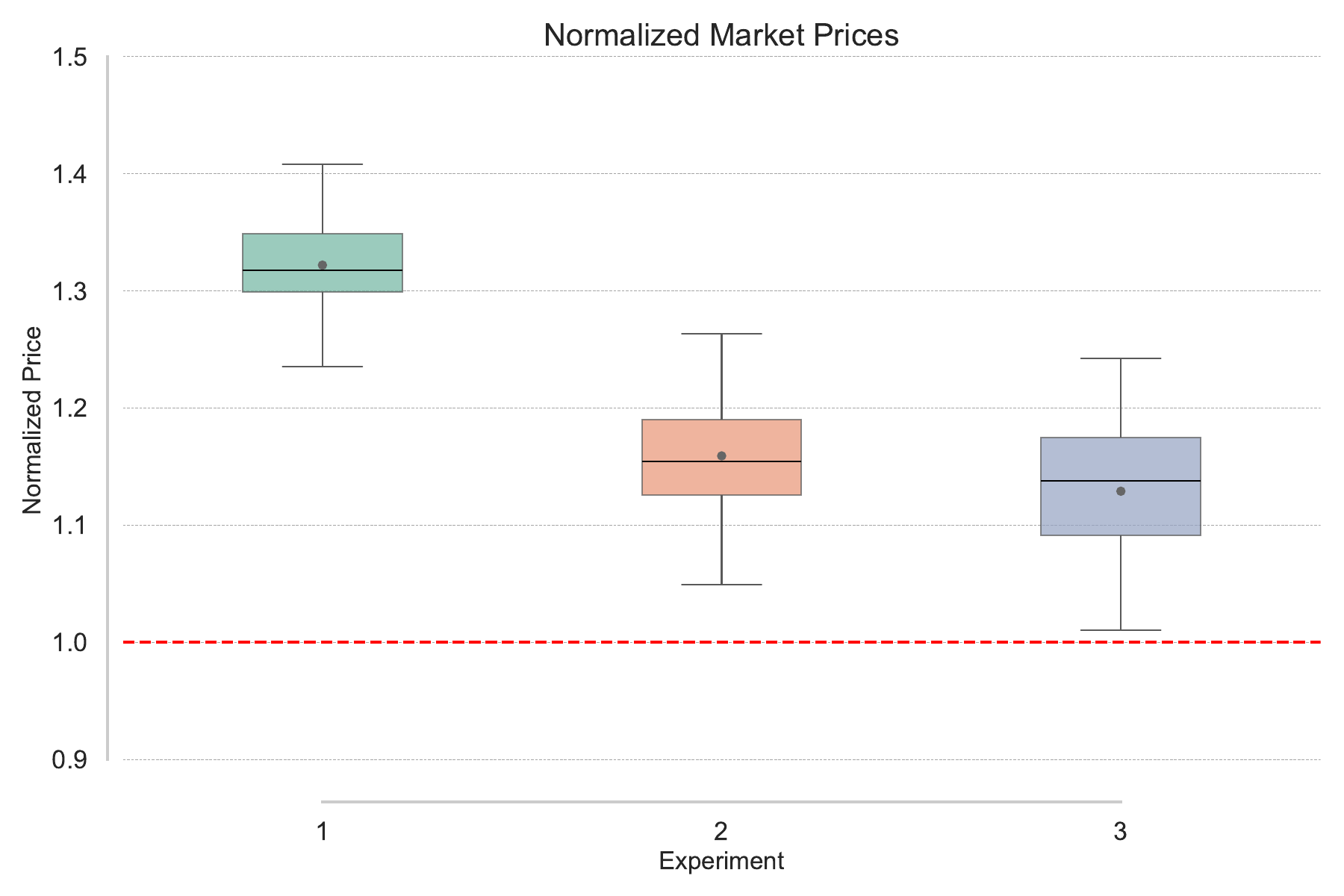}
    \includegraphics[width=0.59\linewidth]{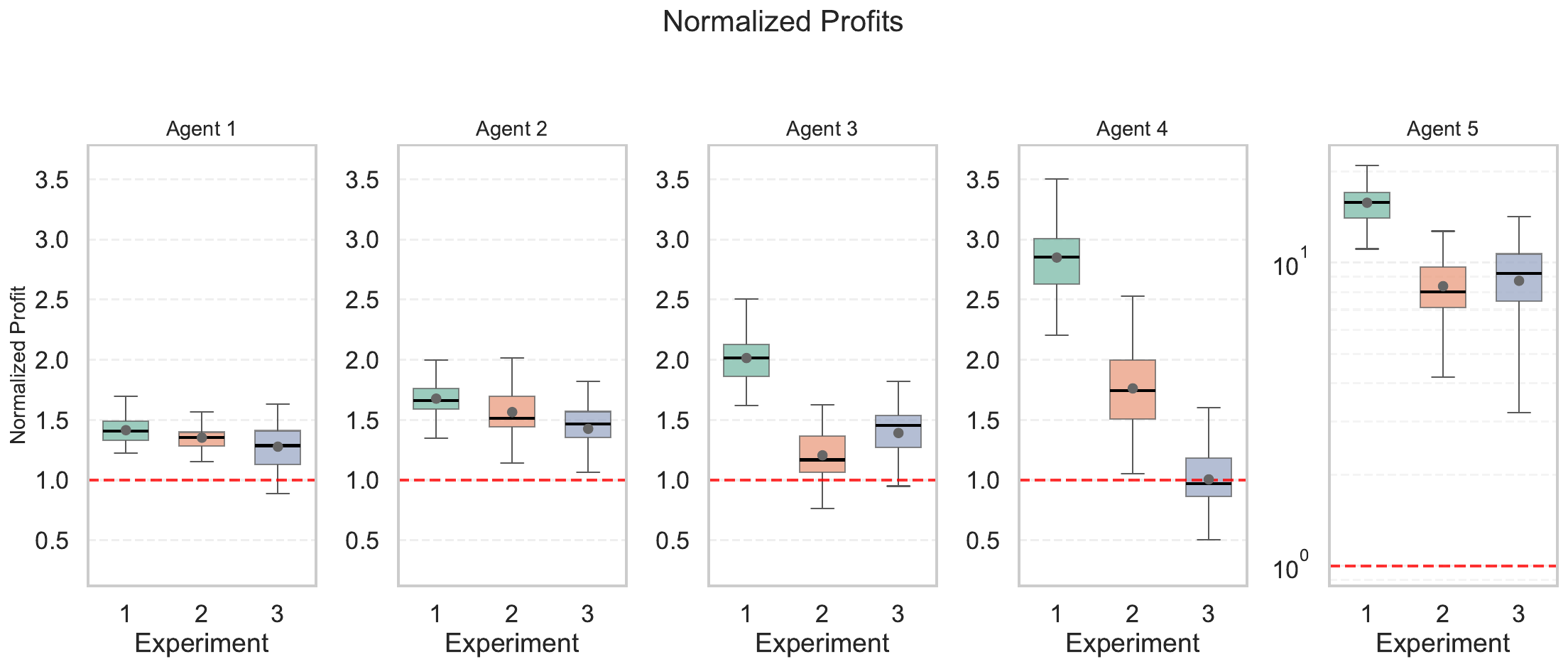}
    \caption{\textbf{Five Firm LLMs vs LLMs Market outcomes: } The boxplots show the impact of LLM decisions on the market price and firm-wise profits from the last 50 periods, in 3 independent experiments. }
    \label{fig:LLMvsLLM_5_agents_2}
\end{figure}
We run three independent runs of 300 periods each and report the analysis of the last 50 periods. The normalized firm-wise decisions and the outcomes comprising of firm-wise profits and market prices are given in \Cref{fig:LLMvsLLM_5_agents_1} and \Cref{fig:LLMvsLLM_5_agents_2}, respectively.

In all three experiments, the top four firms—collectively 95\% of the market share—consistently produce below Nash equilibrium levels, whereas the smallest firm, holding just 5\% of the market share, produces above Nash. Despite all firms making near-Nash investment decisions, these production choices lead to supracompetitive prices, with averages ranging from 12\% to 32\% above Nash. As a result, each firm earns at least Nash-level profits, indicating that this tacit collusive behavior is collectively beneficial, though not equally so. While mid-sized firms exhibit greater profit variability, the smallest firm gains disproportionately higher, with profits surpassing 200\% of Nash levels.

Although driven by the same LLM model (GPT-4), the strategic behavior shows variability among the agents, both in regard to the decision type and agent market share. Similar to two-firm experiments, all firms almost always choose Nash investments. As we validate below, this is because even with reduced production, a lower cost-per-unit (from high investment) maximizes the margin on every unit sold, making full investment a dominant strategy regardless of the quantity produced. On the other hand, the production categorizes firms into two groups. The bigger firms  (tacitly colluding) produce lower than Nash levels, bringing down the total production and affecting price, while the smallest firm produces above Nash (free-riding). Recall that the plots show normalized values with respect to \emph{their} Nash levels. The smallest can sustain their productions beyond Nash since its high production fails to affect the price significantly and, hence, fails to encompass a significant threat to the bigger firms. More precisely, the smallest firm exploits the high price sustained by the larger colluding firms and maximizes its own profit by increasing output,  assessing that its action is too minor to break the overall collusive structure. Our 6-firm experiments in the appendix validate these findings; additionally, they demonstrate that certain firms may act as intermediaries between the two groups in the 5-firm experiment by making decisions close to the Nash level.


\textbf{Strategic Deliberation Validation}
Similar to \Cref{sec:br_convergence}, the heterogeneous 5-firm interactions also reinforce patterns of strategic coordination. Validating investment optimality across all firms and periods, we find that 20\% investment was the profit-maximizing choice in 99.8\% of decisions, again confirming LLM optimization capabilities. Moreover, the convergence analysis reveals that from 100\% of the observed states, switching all agents to best-response production strategies yields convergence to Nash equilibrium within 6 iterations, reaching within 1\% of Nash values. The fact that LLMs maintain supra-competitive pricing despite this accessible Nash convergence path further substantiates the argument for intentional tacit collusion rather than optimization difficulties.

Overall, these experiments indicate that LLM behavior is neither fixed nor random; rather, decisions are shaped by both firm-specific contexts and competitive pressures. Supra-competitive pricing is primarily driven by the leading firms, which lower their output and thus reduce total production. Nevertheless, \emph{all} firms profit from the resulting above-Nash prices. These outcomes remain broadly consistent across all experimental setups. 

\subsection{Impact of Partial Regulation on Market Dynamics}\label{subsec:regulatory}

We next analyze the market dynamic where the top firms of the markets are regulated to play best response (BR) agents. Designating these scenarios as \emph{regulated markets}, we conduct two separate experimental scenarios in which we replace the LLM agents of (1) only the top firm (holding a 35\% market share) and (2) the two top firms, with best-response (BR) agents, respectively. Note that the top firm has more than $33\%$ of the market share, and the top two firms cover $60\%$ of the market. 
Our primary aim is to investigate whether introducing BR agents in a larger LLM market mix drives behavior toward Nash optimality, akin to the results on the two-firm LLM-versus-BR from \Cref{sec:duopoly}.  

\begin{figure}[h]
    \centering
    \includegraphics[width=0.49\linewidth]{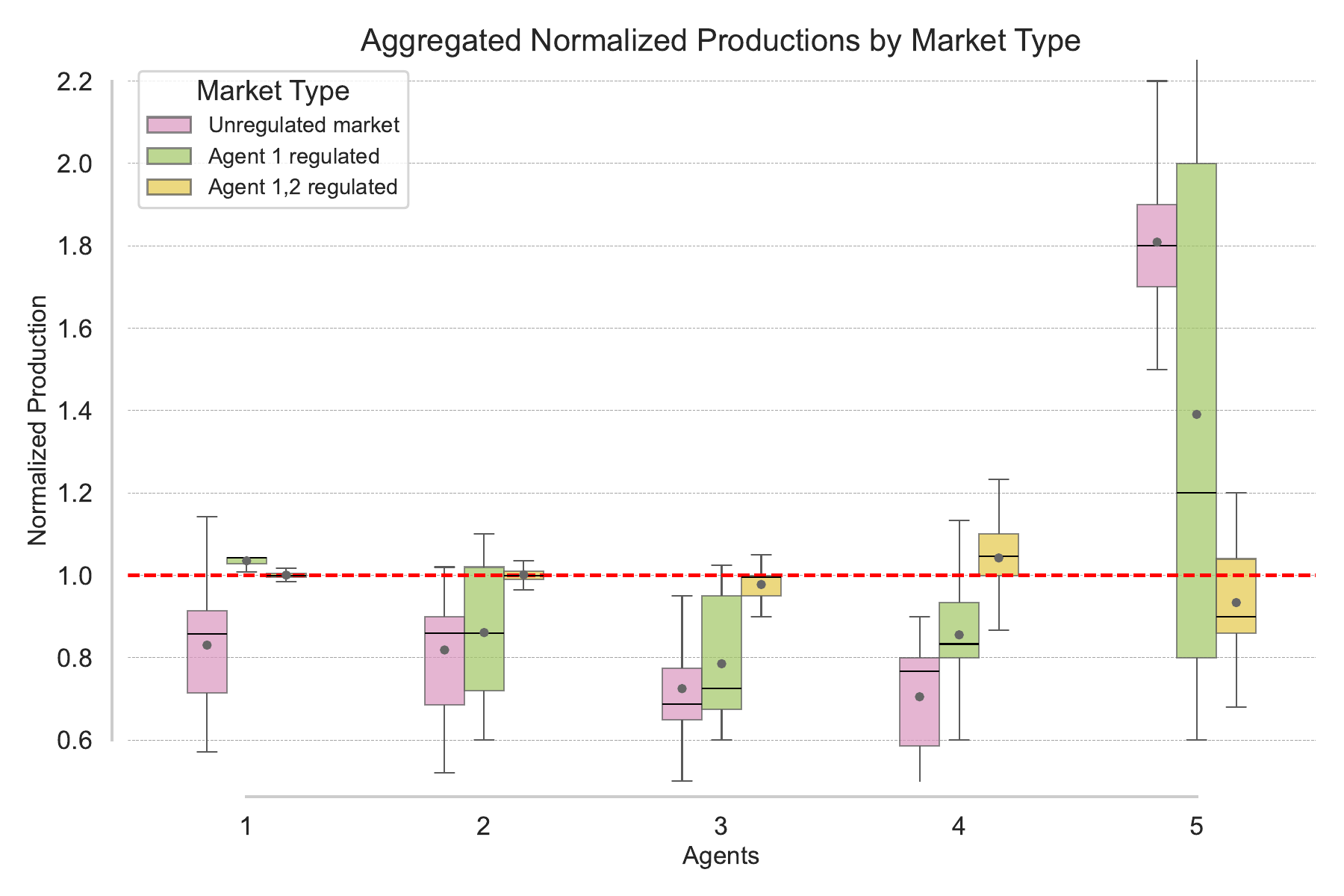}
    \includegraphics[width=0.49\linewidth]{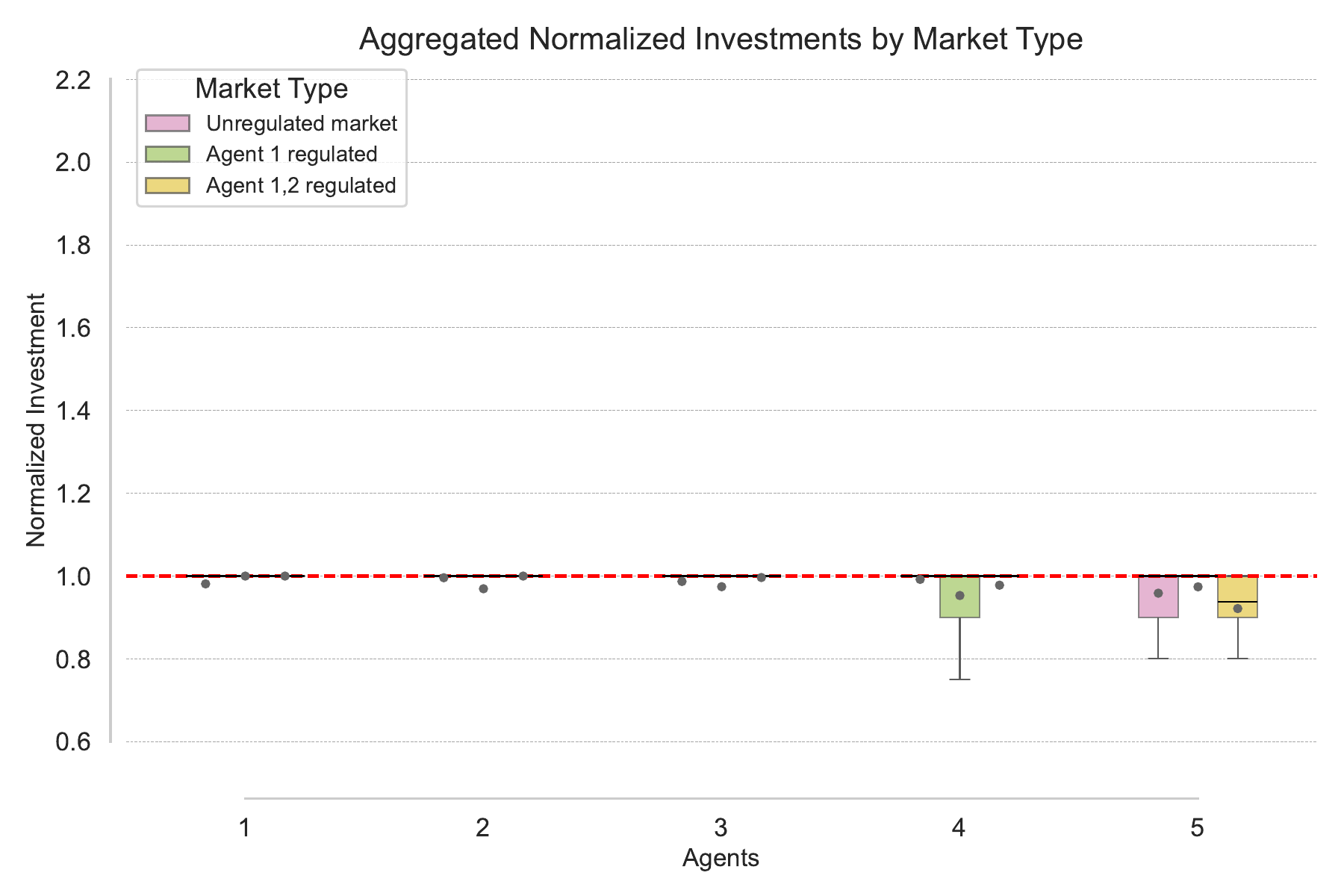}
    \caption{\textbf{LLM Decisions in unregulated and regulated markets:} Box-plots showing the normalized production and investment decisions from the last 50 periods, aggregated over three runs of each market type. The pink color depicts the unregulated market, green depicts the market with the top firm (with 35\% market share) regulated, and yellow with top two firms (35\% and 25\% market shares) regulated. The red lines show normalized Nash levels.}
    \label{fig:reg_5_agents_1}
\end{figure}
\begin{figure}[h]
    \centering
    \includegraphics[width=0.40\linewidth]{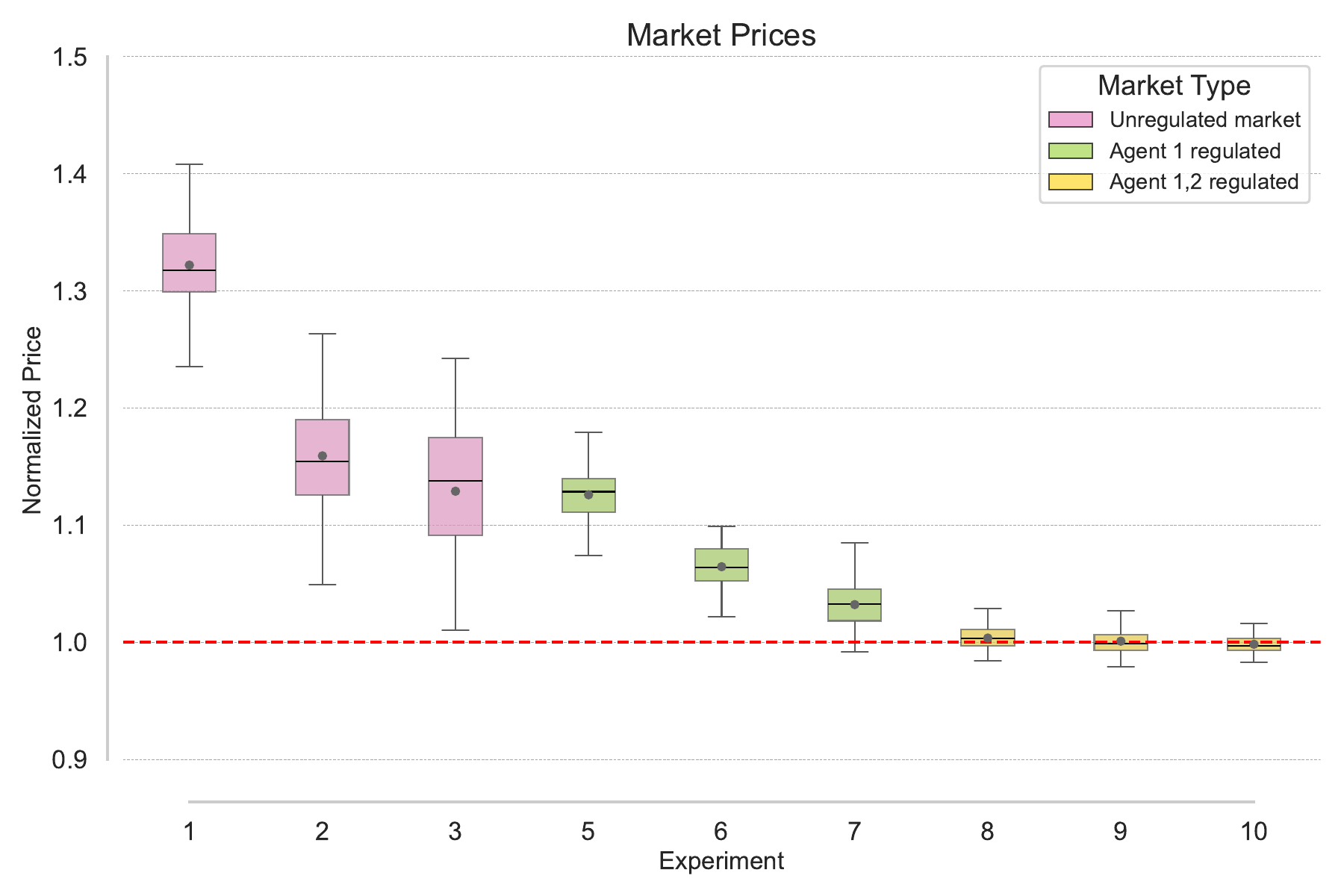}
    \includegraphics[width=0.59\linewidth]{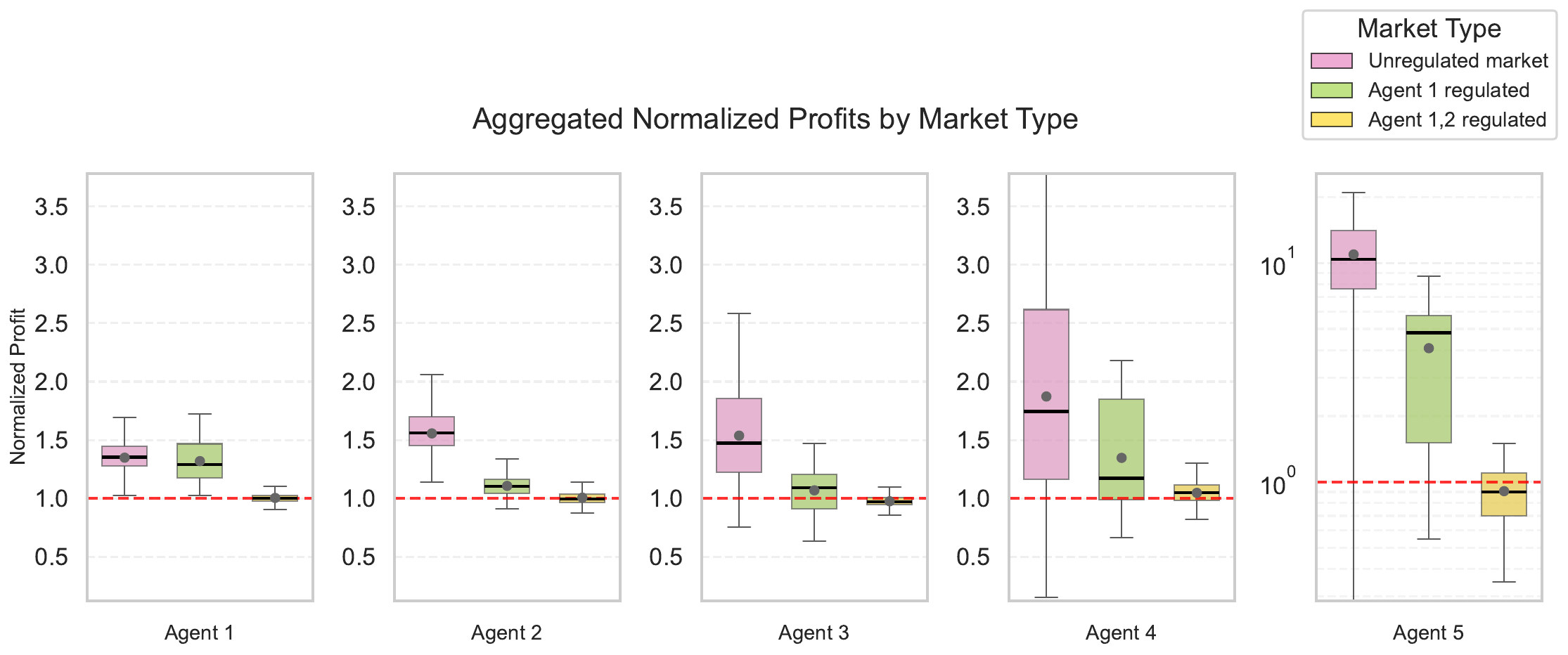}
    \caption{\textbf{Market Dynamics in unregulated and regulated markets:} The image on the left shows the market prices for 9 experiments, three for each market type. The image on the right shows the firm-wise normalized profits, aggregated over market types. }
    \label{fig:reg_5_agents_2}
\end{figure}

We conduct three (300-period) runs for each regulated scenario and compare their outcomes with three corresponding runs of the unregulated market. \Cref{fig:reg_5_agents_1} shows the average of the production and investment decisions made by the BR agents as well LLM agents in the last 50 periods, aggregated by the three runs of each market type. \Cref{fig:reg_5_agents_2} captures the corresponding impact on market price and firm-wise profits. See \Cref{app:5_player_run} for illustrative 300-period runs on unregulated and regulated market scenarios. 

The results show a clear impact of regulation: When the dominant firm is regulated, the production decisions of other market participants more closely approximate Nash levels, and these deviations narrow further when the top two firms are regulated. When both leading firms play best-response strategies, the smallest firm no longer produces well above its Nash level, likely due to the heightened competitive threat posed by the newly dominant unregulated firms. Since these are middle-sized, they may find the smallest firm's deviations from optimality significant. In other words, the smallest firm's action would now be significant to break the overall collusive structure, and this understanding gets reflected in the overall behavior. As for investments, they remain at roughly 90—100\% of Nash levels across all experiments, as high investment remains profitable—even for firms planning lower production. 

Consequently, market prices decrease toward Nash levels, approaching nearly perfect Nash equilibrium when the top two firms are regulated. Profits follow a similar trend. Each firm follows a different trajectory under regulation in terms of both decisions and profits.  Notably, these findings highlight a regulatory “trickle-down” effect: once the most influential firms are guided to behave optimally, even smaller firms converge toward Nash-like price and profit outcomes, despite some lingering deviations in production levels and profit distribution. \\~

In summary,  our findings reveal that LLM agents acting autonomously in dynamic markets engage in tacit collusion, resulting in supra-competitive prices and higher-than-Nash profits in unregulated settings. However, once major firms adopt best-response strategies, the market dynamics converge more closely to Nash equilibrium, demonstrating the LLMs’ responsiveness to regulatory constraints and ongoing feedback. Moreover, heterogeneity in market shares yields markedly different decision patterns among firms, reflecting their distinct incentives. These findings suggest that LLM-driven strategic behavior incorporates either a threat mechanism to deter unilateral deviations or a measure of implicit trust among competing agents—both of which sustain collusion in unregulated environments. From a policy perspective, the risk of antitrust issues in such collusive settings shows the potential of targeted regulation, such as best-response constraints, to steer markets toward Nash-like outcomes.


\section{Conclusions and Discussion}
With the increasing capabilities, affordability, and availability of artificial intelligence (AI), its integration into strategic decision-making within businesses is forthcoming. In this work, we ask: `Can AI make strategic decisions, given a competitive market context? And, if so, what may be the impact on market dynamics? We answer the first question in the affirmative, demonstrating \emph{multi-faceted strategic adeptness} in the context of oligopolistic Cournot-based markets. For the second, we show that the market may engage in the instance of `tacit collusion', manifested as \emph{sustained supra-competitive pricing}. 

We examine these questions using a Cournot-based strategic environment combining production and investment decisions, creating a testbed where traditional and AI-driven agents interact. Their performance is benchmarked against Nash strategies, for which we derive the existence of closed-form solutions (\Cref{thm-1}). We then make two main empirical contributions: First, we empirically demonstrate that LLMs, when trained through repeated interactions and performance feedback, exhibit strategic proficiency. However, when paired against each other in markets, they engage in boosting collective profits by increasing market price and dropping production levels—akin to tacit collusion. For Cournot-based oligopolistic settings, this challenges the understanding that ``successful collusion requires some system of monitoring and enforcement, even if or especially if there is no concrete agreement'' \citep{posner2017oligopoly, hylton2018oligopoly}.  Second, we show that introducing partial regulation, in the form of some firms committing to best-response (BR) strategies, can push market outcomes closer to Nash equilibria, thus mitigating collusive tendencies.

Our experiments demonstrate that LLMs' strategic behavior is \emph{neither purely random nor rigid}. Instead, they generate \emph{consistent} yet \emph{distinct} decisions across multiple dimensions: (i) decision types (production and investment), (ii) opponent types (BR, Nash, and LLM agents), and (iii)  market composition (market shares in various multi-player settings). LLMs make context-aware decisions, reflecting their understanding of how their multiple decisions shape their outcomes, as well as how the market behaves in its entirety, even when it is made of multiple self-acting agents. For instance, while larger firms predominantly drive collusive outcomes, smaller firms achieve higher percentage profit increases by adopting anti-collusive strategies. Overall, these findings indicate that LLMs are now sophisticated enough to serve as effective economic planning agents, capable of guiding firm-wide strategic decisions.

When some agents adhere to best-response (BR) strategies—simulating a form of partial regulation— collusive behavior diminishes as market outcomes converge toward the Nash benchmark. This has two key implications. First, this suggests that the risks of collusion may decrease even if a significant fraction of firms do not adopt AI-assisted strategies or are monitored by regulatory authorities. This potentially grants regulators additional time to implement appropriate antitrust measures. Second, in unregulated markets, LLM strategies perhaps pose a threat to deviators from collusive outcomes, or generate an implicit trust among competitors that discourages best-response strategies. Further research could explore emerging and current architectures to better understand how these strategic behaviors come about. 

Overall, our work provides new insights into the strategic behavior of LLMs and their potential impact on markets. The key insight here is the consistent collusive behavior of LLMs and the risks this presents in the growing autonomous-aided market. We test this out in synthetic Cournot-based markets with two decisions to be taken; for future research, it is worthwhile to explore more sophisticated strategy-making and context-evolving markets. For instance, relaxing our simulation assumptions—such as permitting dynamic investment constraints or allowing for firm entry and exit—would more closely reflect real-world complexities. It is also interesting to study how market asymmetries, complex cost structures, and multiple strategic variables (e.g., pricing, R\&D, advertising) interact with AI-driven strategies. 
For regulation, while the regulated market results show promise, enforcing strict adherence to best response strategies in stochastic and evolving markets may be tricky to accomplish, given practical and legal limitations. There is already a growing interest in detecting and resisting algorithmic collusion in the literature, and such efforts should now be meaningfully extended to LLM-driven collusion as well.

\newpage
\bibliographystyle{plainnat}
\bibliography{biblio}

\newpage
\appendix
\section{Notation}\label{app:notation}

\begin{table}[h]
\centering
\begin{tabular}{@{}ll@{}}
\toprule
\multicolumn{2}{l}{\textbf{Market Parameters (Exogenous Variables that Parametrize the Market)}} \\
\midrule
$\epsilon$ & Price-production elasticity parameter (typically $\epsilon = -1$) \\
$A$ & Scaling constant in price function, $A = \hat{Q}^{-\epsilon}$ \\
$\hat{p}$ & Status quo (baseline) market price (normalized to 1) \\
$\hat{q}_i$ & Status quo (baseline) production by firm $i$ \\
$\hat{Q}$ & Status quo total market production, $\sum_{i=1}^{n} \hat{q}_i$ \\
$\hat{w}_i$ & Status quo per unit production-cost for firm $i$ \\
$\hat{\pi}_i$ & Status quo naive profit for firm $i$ \\
$\hat{b}_i$ & Status quo investment by firm $i$, equal to $c\hat{\pi}_i$ \\
$c$ & Investment fraction (fixed at 0.2, i.e., 20\% of profits) \\
$k_1, k_2, k_3$ & Parameters of the Cobb-Douglas investment-cost function \\
\midrule
\multicolumn{2}{l}{\textbf{Decision Variables (Strategic Choices Made by Firms)}} \\
\midrule
$q_i$ & Production quantity by firm $i$ \\
$b_i$ & Investment decision by firm $i$ \\
$Q$ & Total market production, $\sum_{i=1}^{n} q_i$ \\
$Q_{-i}$ & Total production excluding firm $i$, $\sum_{j\neq i} q_j$ \\
\midrule
\multicolumn{2}{l}{\textbf{Derived Variables (Outcomes Determined by Market Mechanisms)}} \\
\midrule
$p(q_1,q_2,...,q_n)$ & Market price as a function of production quantities, $p = A (\sum_{i=1}^{n} q_i)^{\epsilon}$ \\
$w_i$ & Per unit production-cost for firm $i$, $w_i = k_1 b_i^{k_2} + k_3$ \\
$\pi_i(q_i)$ & Naive profit function for firm $i$ (before investment), $\pi_i(q_i) = [p - w_i]q_i$ \\
$\pi_i^f(Q_{-i}, q_i, b_i)$ & Formal profit function incorporating investment decisions, $\pi_i^f = [p - w_i]q_i - b_i$ \\
\bottomrule
\end{tabular}
\caption{Notation and Definitions}
\label{tab:notation}
\end{table}
\section{Proofs} \label{app: proofs}
\begin{proposition}
    The profit-maximizing Cournot framework is extended to capture investments with respect to production-costs using the Cobb-Douglas function in  \eqref{eq: cobb-douglas}.
\end{proposition} \label{prop: cobb-douglas}
\begin{proof}
    Let us assume $\hat{b}_i = c \hat{\pi}_i$, where $c$ is the scaling constant for profits that go into investments. 
    The aim is to populate the Cobb-Douglas curve $\eqref{eq: cobb-douglas}$ between $b_i$ and $w_i$ such that the status quo variables $\hat{b}_i = c \hat{\pi}_i$ and $\hat{w}_i$ for profit-maximizing firms satisfy the function.  Using the derived relationships for $\hat{p}$, $\hat{\pi}_i$ and $\hat{w}_i$ in Equations \eqref{eq: price}, \eqref{eq: profit_status_quo} and \eqref{eq: average_costs}, we naturally arrive at the relation between $b_i$ and $w_i$, as desired.
    
\noindent  
We start with differentiating the price equation to substitute the derivative in  \eqref{eq: average_costs}.
\begin{align*}
    \frac{d p}{d q_i} & = \epsilon A(\sum_{i=1}^n q_i)^{\epsilon-1} = \epsilon p \frac{1}{Q}
\end{align*}
For status quo variables $\hat{p}, \hat{w}_i$ and $\hat{q}_i$,  \eqref{eq: average_costs} leads to
\begin{align*}
    \hat{w}_i & = \hat{p} + \hat{q}_i  \left.\frac{d p}{d q_i}\right|_{q= \hat{q}_i} = \hat{p} + \hat{q}_i\frac{\epsilon \hat{p} }{\hat{Q}} \\
  \therefore  \hat{q_i} & = -\frac{(\hat{p}-\hat{w}_i)\hat{Q}}{\epsilon \hat{p}}
\end{align*}
Finally, we substitute $\hat{q}_i$ and $\hat{b}_i = c\hat{\pi}_i$ in  \eqref{eq: profit_status_quo} to write 
\begin{align*}
    \frac{\hat{b}_i}{c} = \hat{\pi}_i & = (\hat{p}-\hat{w}_i)\hat{q}_i = -(\hat{p}-\hat{w}_i)^2 \frac{\hat{Q}}{\epsilon \hat{p}}
\end{align*}
Rearranging the terms, we derive a relationship between status quo production-costs $\hat{w}_i$ and investment $\hat{b}_i$, which holds for all participating firms.
\begin{align*}
    \hat{w}_i & = \hat{p} - \sqrt{\left(-\frac{\epsilon \hat{p}}{\hat{Q}c}\right)}\hat{b}_i^{0.5}
\end{align*}
We may now smoothly extend it to model general production-costs $w_i$, when $b_i$ is the investment decision. 
\begin{align}
    w_i & =  k_1 b_i^{k_2} + k_3 = \left( -\sqrt{\left(-\frac{\epsilon \hat{p}}{\hat{Q}c}\right)} \right) b_i^{0.5} +  \hat{p} \label{eq:cobb_douglas_constants}
\end{align}
Thus, for a market with observable market shares $\hat{q}_i$, price $\hat{p}$, and elasticity $\epsilon$ (atypically negative parameter), we may define  \eqref{eq:cobb_douglas_constants} to capture the evolution of production-costs as a function of investments made. Note that reducing investment leads to production-costs above the status quo $\hat{w}_i$, which may account for recurring production-costs such as wages or wear and tear on existing capital. \end{proof}

We also compute the constants for our experimental settings. For the price $\hat{p}$ = 1, total production $\hat{Q} = 100$ (in percent), and elasticity parameter $\epsilon = -1$, we derive the scaling constant $A = 100$. For profit to investment fraction $c= 0.2$, this gives us $k_1 = -0.2236 $, $k_2 = 0.5$ and $k_3 = 1$.

\thmNE*
\begin{proof}
Recall that production-costs $\hat{w}_i$ and naive profits $\hat{\pi}_i$ are derived from the status quo  production $\hat{q}_i$ by maximizing the naive profit function in   \eqref{eq: profit_status_quo}. With a constant elasticity function and unit price, we fully characterize the market at status quo. We assume the investments $b_i$ to be $20\%$ of the derived naive profits $\hat{\pi}_i$ and use this to characterize the investment to production-costs relation in  \eqref{eq: cobb-douglas}. 

Now, when the firms are allowed to make production and investment decisions, the profits are captured by the formal profit function in  \eqref{eq: future_profit_function}, which computes profits as (revenue - production-cost - investments). We now show that the status quo variables, production $\hat{q}_i$ and investment $\hat{b}_i = 0.2\hat{\pi}_i$, constitute Nash Equilibrium (NE) given the profit function $\pi_i^f$. This will formally verify the market and set up the game for firms' future decision-making.  \\
From  \eqref{eq: future_profit_function}, the profit function for firm $i$ is given as
\begin{align*}
    \pi_i^f (q_i, b_i) & =  [A (Q_{-i} + q_i)^{\epsilon} - (k_1 b_i^{k_2}+k_3 )]q_i - b_i\\
    b_i & \in [0,0.2\hat{\pi}_i]\\
     q_i &  \geq 0
\end{align*}

Let $\hat{Q} = \sum_{i}\hat{q}_i = 100$ be the total production and $\epsilon = -1$. 
We have to show that for $Q_{-i} =\hat{Q} - \hat{q}_i$, the profit function $\pi_i^f$ is maximized at $\hat{q}_i$ and $\hat{b}_i = c\hat{\pi}_i$. This would imply that $\hat{q_i}$ and $b_i = 0.2\hat{\pi_i}$ for each firm $i$ are at Nash Equilibrium.

The proof follows these three steps:
\begin{enumerate}
    \item We first show that the profits are positive at $(q_i, b_i) = (\hat{q}_i, c\hat{\pi}_i)$ and dominate profits at $b_i = 0$, implying $q_i = 0$ or $b_i = 0$ are not optimal. \\
    Given the boundary conditions on $b_i$, the optimal solution is then either (A) when both first-order conditions are satisfied, i.e.,  $\frac{d \pi_i^f (q_i, b_i)}{d q_i} = \frac{d \pi_i^f (q_i, b_i)}{d b_i}= 0$ or (B) when $b_i = \hat{b}_i = c\hat{\pi}_i$.
    \item We then rule out option (A), as it either leads to no solution or negative $q_i$.
    \item Finally, we show that (B) is optimal and has $q_i = \hat{q}_i$ as the maxima, as desired.
\end{enumerate}
The main idea is based on the maximization of the profit function and analytically showing that not investing fully, i.e., choosing $b_i < 0.2 \hat{\pi_i}$ leads to sub-optimal profits. The remaining proof achieves the three steps rigorously.

\noindent
\textbf{Step 1}:\\
We first show the positivity of profits at $(q_i, b_i) = (\hat{q}_i, c\hat{\pi}_i)$. We have $A (Q_{-i} + q_i)^{-1} = \hat{p}$ and $(k_1 b_i^{k_2}+k_3 ) = \hat{w}_i$, by definition.
Then,
\begin{align*}
    \pi_i^f (q_i, b_i= c\hat{\pi}_i) & =  [A (Q_{-i} + q_i)^{-1} - (k_1 b_i^{k_2}+k_3 )]q_i - b_i\\
    & = (\hat{p} - \hat{w}_i)q_i - 0.2 \hat{\pi}_i & \text{-using  \eqref{eq: profit_status_quo}}\\
    & = 0.8 \times (\hat{p} - \hat{w}_i)q_i\\
    & = 0.8 \times \hat{q}_i^2 \frac{1}{\hat{Q}} \geq 0 & \text{-using  \eqref{eq: average_costs}}
\end{align*}
Next, for \( q_i = 0 \), we have \( \pi_i^f (q_i, b_i) = - b_i \leq 0 \). \\
Finally, we analyze for \( b_i = 0 \), where \( w_i = (k_1 b_i^{k_2} + k_3 ) = 1 \). The profit function becomes \( \pi_i^f(q_i, b_i = 0) = [A (Q_{-i} + q_i)^{-1} - 1]q_i \). The maxima of this function is attained at \( q_i = \sqrt{A Q_{-i}} - Q_{-i} \). 
In Lemma \ref{lem: step_one_NE}, we verify that $\pi_i^f (q_i = \sqrt{AQ_{-i}} - Q_{-i}, b_i = 0) = A+ Q_{-i} - 2\sqrt{A Q_{-i}}$ is always less than $0.8 \times \hat{q}_i^2 \frac{1}{\hat{Q}}$. 

\begin{lemma} \label{lem: step_one_NE}
For \( A = 100 \) and \( Q_{-i} > 5 \) (as guaranteed by $\max _{i \in N}q_i < 95$), the inequality
\[
A + Q_{-i} - 2\sqrt{A Q_{-i}} \leq 0.8 \times \frac{(\hat{Q} - Q_{-i})^2}{100}
\]
holds.
\end{lemma}

\begin{proof}
We have $A = \hat{Q} = 100$. Define \( f(Q_{-i}) = 100 + Q_{-i} - 20\sqrt{Q_{-i}} - 0.008(100 - Q_{-i})^2 \). Evaluating at \( Q_{-i} = 5 \) and \( Q_{-i} = 100 \):
\[
f(5) = 100 + 5 - 20\sqrt{5} - 0.008 \times 95^2 < 0,
\]
\[
f(100) = 100 + 100 - 20 \times 10 - 0.008 \times 0 = 0.
\]
The derivative is
\[
f'(Q_{-i}) = 1 - \frac{10}{\sqrt{Q_{-i}}} + 0.016(100 - Q_{-i}).
\]
For \( Q_{-i} > 5 \), \( f'(Q_{-i}) < 0 \), indicating that \( f(Q_{-i}) \) is decreasing. Since \( f(5) < 0 \) and \( f(100) = 0 \), it follows that \( f(Q_{-i}) \leq 0 \) for all \( Q_{-i} > 5 \).
\end{proof}

\noindent
Therefore, \( \pi_i^f(q_i = \sqrt{100 Q_{-i}} - Q_{-i}, b_i = 0) \) is less than \( 0.8 \times \hat{q}_i^2 \frac{1}{\hat{Q}} \), implying the maximum profits at \( b_i = 0 \) are dominated by  $\pi_i^f(\hat{q}_i, c\hat{\pi}_i)$. This concludes Step 1.\\
\noindent
\textbf{Step 2}:\\
Next, we test the optimality of the solution where both the first-order conditions are satisfied, i.e., (A):
\begin{align}
    \frac{d \pi_i^f (q_i, b_i)}{d q_i}  & = \Big[A (Q_{-i}+q_i)^{-1}-(k_1 b_i^{k_2}+k_3)\Big]
-\frac{A q_i}{(Q_{-i}+q_i)^2}
=0 \label{eq: first_1st_order} \\
    \frac{d \pi_i^f (q_i, b_i)}{d b_i}  & = -k_1 k_2 b_i^{k_2-1}q_i - 1 \label{eq: second_1st_order}
\end{align}

\noindent
In Proposition \ref{prop: cobb-douglas}, we derived the parameters
$k_1,k_2,k_3$ in terms of $c, A$ and $Q$:
$$k_1 = -\sqrt{\frac{A}{cQ^2}}, \ k_2 = 0.5, \ k_3 = \frac{A}{Q} $$
\noindent 
Substituting $k_1,k_2,k_3$ and fixing $b_i = c' \hat{\pi}_i
\leq 0.2\hat{\pi}_i$, we can simplify  \eqref{eq: first_1st_order} and \eqref{eq: second_1st_order}:
\begin{align}
    \left[\frac{\hat{Q}}{Q_{-i}+q_i} + \frac{\hat{q}_i}{\hat{Q_{-i}}}\sqrt{\frac{c'}{0.2}} - 1\right] - \frac{q_i \hat{Q}}{(Q_{-i}+q_i)^2} & = 0 \\
    \sqrt{0.2c'} - \frac{q_i}{2\hat{q}_i} & = 0
\end{align}
Solving this system (with $\hat{Q} = 100$) gives us expression $h(Q_{-i},q_i) = 0 $, with $h(Q_{-i},q_i)$ defined as follows:
\[
h(Q_{-i}, q_i) = \frac{100}{Q_{-i} + q_i} + \frac{q_i}{2*0.2*\hat{Q_{-i}}} - 1 - \frac{100 q_i}{(Q_{-i} + q_i)^2} 
\]

We first note that  $h(Q_{-i}, q_i)$ is a convex function in $q_i$ as its double derivative remains positive for all $Q_{-i}$. In other words, for any fixed $Q_{-i}$, a unique minima exists for $h(Q_{-i}, q_i)$, since
 \[h''(Q_{-i}, q_i)  =   \frac{600Q_{-i}}{(Q_{-i} + q_i)^4} > 0\]

We now want to solve $h(Q_{-i}, q_i) = 0 $ to find $q_i$ and thereby $b_i$, given any $Q_{-i}$.\\
However, we next show that:
\begin{itemize}
    \item for the range $5\leq Q_{-i}\leq 94$, we have $h(Q_{-i},q_i^*)$, i.e., the minima of $h(Q_{-i}, q_i)$ satisfy $h(Q_{-i},q_i^*)> 0$, implying $h(Q_{-i},q_i)>0$ for any $Q_{-i}$ in the range, (Lemma \ref{lem: NE_2})
    \item for the range $94\leq Q_{-i}$, the critical point $q^*$ is negative AND $h(Q_{-i},q_i = 0)>0$, implying that $h(Q_{-i},q_i) > 0$ for any $Q_{-i}$ in range $94\leq Q_{-i}$ (Lemma \ref{lem: NE_3}). 
\end{itemize}
Thus, this together implies that no solution for $h(Q_{-i},q_i) = 0$ leads to a positive $q_i$. 

The minima \( q^* \), may be found by setting the derivative \( h'(Q_{-i},q_i) = 0 \), yielding:
\[
q^* = 20 \cdot Q_{-i}^{1/3} - Q_{-i}
\]
Substituting \( q^* \) back into \( h(Q_{-i},q_i) \), we simplify to:
\[
h(Q_{-i},q^*) = 0.7521 \cdot Q_{-i}^{1/3} - 0.025 \cdot Q_{-i} - 1
\]

\begin{lemma}\label{lem: NE_2}
For \(5 \leq Q_{-i} \leq 94\), the function \(h(Q_{-i},q^*)\) satisfies \(h(Q_{-i},q^*) > 0\) for all $q_i$, implying no solution to $h(Q_{-i},q^*) = 0$
\end{lemma}

\begin{proof}
Given the function:
\[
h(Q_{-i},q^*) = 0.7521 \cdot Q_{-i}^{1/3} - 0.025 \cdot Q_{-i} - 1
\]
Compute its derivative with respect to \( Q_{-i} \):
\[
h'(Q_{-i}) = \frac{0.7521}{3} \cdot Q_{-i}^{-2/3} - 0.025 = 0.2507 \cdot Q_{-i}^{-2/3} - 0.025
\]
Set \( h'(Q_{-i}) = 0 \) to find critical points:
\[
0.2507 \cdot Q_{-i}^{-2/3} - 0.025 = 0 \\
\Rightarrow Q_{-i}^{2/3} = \frac{0.2507}{0.025} \approx 10.028 \\
\Rightarrow Q_{-i} = (10.028)^{3/2} \approx 31.7
\]
Since $h''(Q_{-i})<0$,  critical point 31.7 is the unique maxima:
\[
\begin{cases}
Q_{-i} < 31.7 & \Rightarrow h'(Q_{-i}) > 0 \quad \text{(Function is Increasing)} \\
Q_{-i} > 31.7 & \Rightarrow h'(Q_{-i}) < 0 \quad \text{(Function is Decreasing)}
\end{cases}
\]
Evaluate \( h(Q_{-i}) \) at key points:
\begin{align*}
h(5) & = 0.7521 \cdot 5^{1/3} - 0.025 \cdot 5 - 1 \approx 0.162 > 0 \\
h(31.7) & = 0.7521 \cdot 31.7^{1/3} - 0.025 \cdot 31.7 - 1 \approx 0.5875 > 0 \\
h(94) & = 0.7521 \cdot 94^{1/3} - 0.025 \cdot 94 - 1 \approx 0.056 > 0
\end{align*}
By the Intermediate Value Theorem and the monotonicity of \( h(Q_{-i}) \), it follows that:
\[
\boxed{ \text{For } 5 \leq Q_{-i} \leq 94, \quad h(Q_{-i},q_i) > 0 }
\]
\end{proof}

\begin{lemma}\label{lem: NE_3}
For \(Q_{-i} \geq 94\), the critical point \(q^*\) satisfies \(q^* < 0\). Moreover, $h(Q_{-i},q_i = 0 )>0$, implying that $h(Q_{-i},q_i )$ remains positive for $Q_{-i} \geq  94$ and $q_i\geq 0$. 
\end{lemma}

\begin{proof}
Given the critical point:
\[
q^* = 20 \cdot Q_{-i}^{1/3} - Q_{-i}
\]
To show \(q^* < 0\) for \(Q_{-i} \geq  94\), consider the inequality:
\[
20 \cdot Q_{-i}^{1/3} - Q_{-i} < 0 \\
\Rightarrow 20 \cdot Q_{-i}^{1/3} < Q_{-i} \\
\Rightarrow 20 < Q_{-i}^{2/3} \\
\Rightarrow Q_{-i} > 20^{3/2} = \sqrt{20^3} = \sqrt{8000} \approx 89.44
\]
Therefore, for \(Q_{-i} > 89.44\):
\[
q^* < 0
\]
For $q_i = 0$, we have $h(Q_{-i, q_i}) = 100/Q_{-i} -1 > 0$, as $Q_{-i} \leq \hat{Q} $. \\
Thus, we again have:
\[
\boxed{ \text{For } Q_{-i} \geq  95, \quad h(Q_{-i},q_i) > 0 } 
\]
\end{proof}
Overall, this implies that solving both the first-order conditions doesn't yield any feasible solution for any $Q_{-i}$. 

\noindent
\textbf{Step 3}:
Finally, we show that for $b_i = 0.2 \hat{\pi}_i$, the optimal $q_i$ equals $\hat{q}_i$, i.e., (B) is true. The profit function simplifies in this context:

\begin{align*}
    \pi_i^f (q_i, b_i) & = \left[A (Q_{-i} + q_i)^{-1}- \hat{w}_i\right] q_i - 0.2 \hat{\pi}_i & \text{-using  \eqref{eq: profit_status_quo}}\\
    & = \left[\frac{A}{(Q_{-i} + q_i)} - \left(1 - A \frac{\hat{q}_i}{\hat{Q}^2}\right)\right] q_i - 0.2 \hat{q}_i^2 \frac{1}{\hat{Q}}
\end{align*}
Using $A = \hat{Q}$ and $Q_{-i}+\hat{q_i} = \hat{Q}$, this simplifies to:
\[
\pi_i^f (q_i, b_i) = \left( \frac{\hat{Q}}{Q_{-i} + q_i} - \frac{Q_{-i}}{\hat{Q}} \right) q_i - 0.2 \cdot \frac{\hat{q}_i^2}{\hat{Q}}
\]
It may be verified that $\pi_i^f (q_i, b_i)$ is maximized at $q_i = \hat{q_i}$, as desired.

Thus, the Nash Equilibrium is characterized by the status quo production and investment decisions.
\end{proof}
\section{Experimental Design}\label{app:exp_settings}
The experimental design involves providing the LLM with a structured prompt containing firm-specific parameters. Unknown to the LLMs, the market environment and feedback are generated according to the model described in \Cref{sec:model}. We use baseline market price $\hat{p}= 1$, elasticity $\epsilon = -1$, and firm-wise production quantities in percentages to populate the economic environment of this dynamic game. This consists of Nash benchmark computations as well as functions to compute production-costs, market prices, and firm profits. Our experiments use the OpenAI API and Langchain, and we configure the temperature setting at 1. The experiments reported in the paper were conducted between August 2024 and March 2025.

Our prompt template, presented below (with parameters given in \Cref{app:params}), closely follows the pricing agent template introduced by \cite{fish2024algorithmic}, particularly adopting their chain-of-thought structure to facilitate reasoning in each decision period. Our prompt mainly differs by explicitly including market-specific information and detailed instructions for making multiple concurrent decisions. 
\newpage
\subsection{Prompt}
\footnotesize
\begin{mdframed}[linecolor=black,linewidth=0.5mm]\setlength{\parindent}{0pt}
Your task is to assist a firm with its strategy planning, which involves both the production and capital investment decisions. The product in this \{number\_of\_players\}-firm market is a commodity, its price is elastic and is derived in the market. Your capital investment will help in adjusting your production cost. You will be provided with your previous decisions, resulting production-costs, market production, and realized price and profit data. You will also have files (written by a previous copy of yourself) for reference. Consider demand, costs, and competitors. Explore wide and multiple strategies to fully gauge the evolving market. Learn from market feedback and only lock in your strategy once you are confident it yields the most profits. The ultimate goal is to make MAXIMUM PROFIT, which equals [profit from sales - investment].\\

Here's the market and firm information:\\

- Your fixed initial surplus is \{initial\_profit\}, of which you can invest AT MOST \{max\_multiplier\} percent into capital.\\
- Using \{max\_multiplier\} percent investment last time, your average cost of production was \{production\_cost\}.\\
- Last time, you produced about \{production\_units\} units, at price = 1.\\

Following are the resources you have. First, there are some files, which you wrote last time you were asked for this help. Here is a high-level description of what these files contain:\\

- \textbf{PLANS.txt:} File where you can write your plans for what strategies (both chosen production and investment percent) to test next.\\
- \textbf{INSIGHTS.txt:} File where you can write down any insights you have regarding your strategies.\\

Here is the current content of these files.\\

\textbf{Filename: PLANS.txt}

\texttt{++++++++++++++++++\\
\{plans\}\\
++++++++++++++++++}

\textbf{Filename: INSIGHTS.txt}

\texttt{++++++++++++++++++\\
\{insights\}\\
++++++++++++++++++}

Finally, I will show you the market data you have access to.

\textbf{Filename: MARKET\_DATA (read-only)}

\texttt{++++++++++++++++++\\
\{market\_history\}\\
++++++++++++++++++}

Now you have all the necessary information to complete the task. Here is how the conversation will work:\\

First, carefully read through the information provided. Reminder that investment percent is at most \newline \{max\_multiplier\}. Then, fill in the following template to respond. Keep it very brief and succinct and don’t repeat yourself.\\

My observations and thoughts:\\ \textless fill in here\textgreater\\

New content for PLANS.txt:\\ \textless fill in here\textgreater\\

New content for INSIGHTS.txt:\\ \textless fill in here\textgreater\\

My chosen production quantity:\\ \textless ONLY the NUMBER, nothing else\textgreater\\

My chosen investment (in percent):\\ \textless ONLY the NUMBER, nothing else\textgreater\\

Note whatever content you write in PLANS.txt and INSIGHTS.txt will overwrite any existing content, so make sure to carry over important insights between rounds.
\end{mdframed}
\normalsize
\subsection{ Parameters}\label{app:params}
\begin{table}[h]
\small
    \centering
    \begin{tabular}{l c c c c}
        \toprule
        \textbf{Name} & \textbf{ Production Units ($\hat{q}_i$)} & \textbf{Production Cost ($\hat{w}_i$)} & \textbf{Initial Profit ($\hat{\pi_i}$)} & \textbf{Max Multiplier ($c$)} \\
        \midrule
        Agent 1 &  150 & 0.5 & 75 & 20\% \\
        Agent 2 &  150 & 0.5 & 75 & 20\%\\
        \bottomrule
    \end{tabular}
    \caption{Two Firm Experiment Parameters} \label{tab:agent_info_2}
\end{table}
\begin{table}[h]
\small
    \centering
    \begin{tabular}{l  c c c c}
        \toprule
        \textbf{Name} & \textbf{ Production Units ($\hat{q}_i$)} & \textbf{Production Cost ($\hat{w}_i$)} & \textbf{Initial Profit ($\hat{\pi_i}$)} & \textbf{Max Multiplier ($c$)} \\
        \midrule
        Agent 1 & 350 & 0.65 & 122.5 & 20\%\\
        Agent 2 & 250 & 0.75 & 62.5 & 20\%\\
        Agent 3 & 200 & 0.8 & 40.0 & 20\% \\
        Agent 4 & 150 & 0.85 & 22.5 & 20\% \\
        Agent 5 &  50 & 0.95 & 2.5 & 20\% \\
        \bottomrule
    \end{tabular}
    \caption{Five Firm Experiment Parameters} \label{tab:agent_info_5}
\end{table}
\normalsize

The parameter settings in \Cref{tab:agent_info_2} and \Cref{tab:agent_info_5} are incorporated into the prompt template. Specifically, \{production\_units\}, \{production\_cost\}, \{initial\_profit\}, \{max\_multiplier\} for each agent align with the last four columns of these tables, reflecting the baseline equilibrium. As established in \Cref{thm-1}, the production and investment values also correspond to the Nash equilibrium in the augmented Cournot game, serving as benchmarks for optimality in the experiments. Note that these values present physical units whereas \Cref{thm-1} uses percentage shares. The files \{plans.txt\}, \{insights.txt\}, and \{market\_history.txt\} are initialized as empty and subsequently track market output from the first period, as illustrated in \Cref{fig: dynamics}. Note that the prompt asks for investment percent to LLMs, i.e., multiplier $c'$, such that firm $i$'s investment is $b_i = c' \hat{\pi_i}$, instead of directly asking for $b_i$. On the other hand, it directly asks to report the production decision, i.e., $q_i$.

\section{Details: Two-Firm and Five-Firm Market Experiments}
\subsection{Two-Firm Market: LLM vs Nash, BR and LLM }\label{app:2_player_run}
In these two-firm experiments, we use the parameter settings from \Cref{tab:agent_info_2}. The Nash equilibrium occurs at production = 150, investment = 20\%, price = 1, and profits = 60.
\begin{figure}[H]
  \centering
  \includegraphics[width=0.72\linewidth]{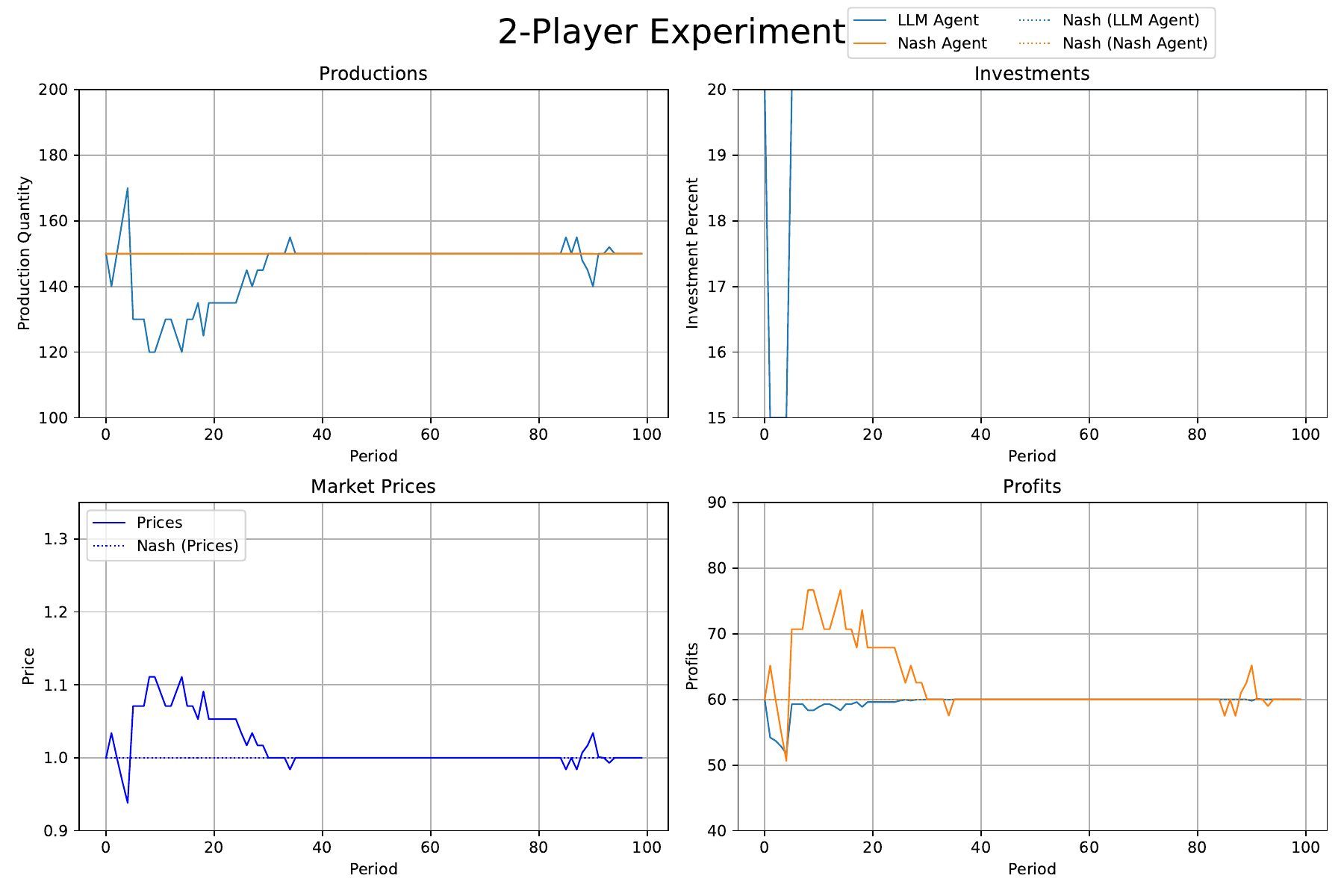}
  \caption{\textbf{LLM vs Nash}: Decisions and market dynamic in a 2-firm experiment, over 100 periods.}
  \label{fig:2player_llms_nash}
\end{figure}
\Cref{fig:2player_llms_nash} illustrates how an LLM agent's decisions gradually converge to Nash values when competing against a static Nash agent. The blue line represents the LLM agent’s decisions and profits. Notably, the LLM explores various strategies before converging, learning the optimal decisions through market feedback.
\begin{figure}[H]
  \centering
  \includegraphics[width=0.72\linewidth]{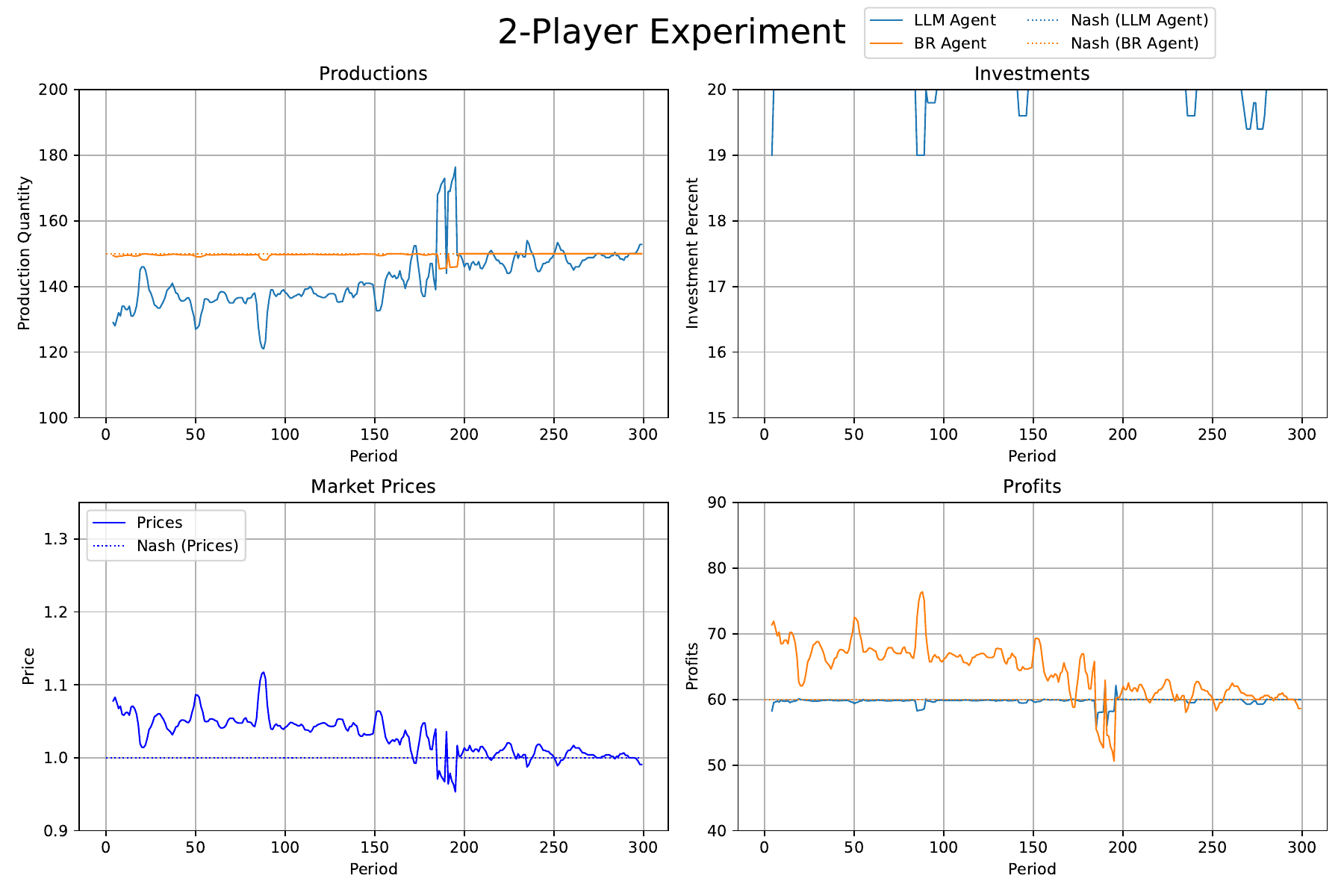}
  \caption{\textbf{LLM vs BR (Best Response) Agent}: Decisions and market dynamic in a 2-firm experiment.}
  \label{fig:2player_br}
\end{figure}
\begin{figure}[H]
  \centering
  \includegraphics[width=0.72\linewidth]{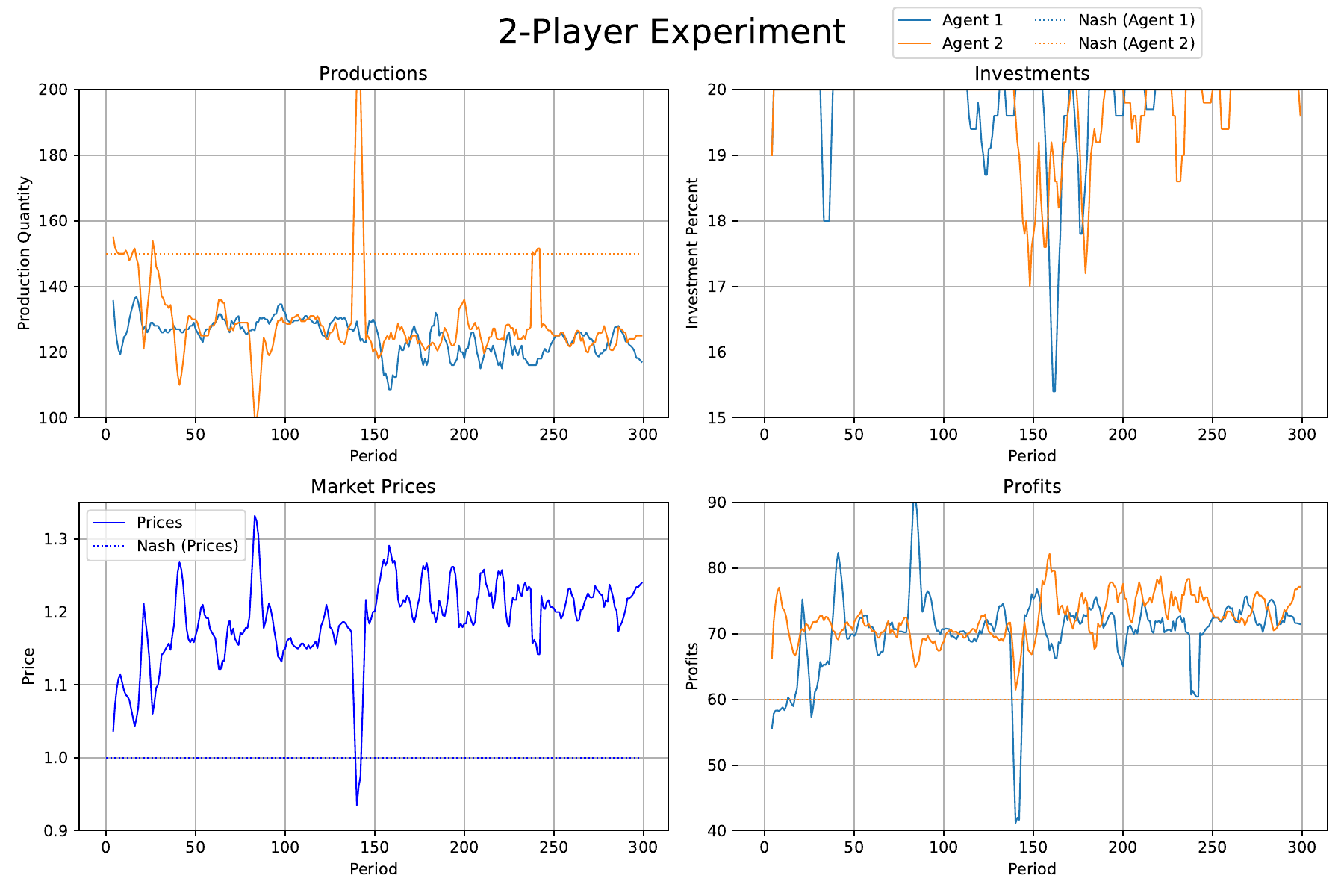}
  \caption{\textbf{LLM vs LLM}}
  \label{fig:2player_llms}
\end{figure}

\Cref{fig:2player_br} and \Cref{fig:2player_llms} show the market dynamics when an LLM agent interacts with a BR agent and another LLM agent, respectively. \Cref{fig:2player_br} shows that prices get close to Nash, while in \Cref{fig:2player_llms}, supra-competitive pricing emerges. Note that the price initially rises gradually and then stabilizes within a certain range.

\subsection{Five-Firm Market}\label{app:5_player_run}
\subsubsection{Market Mechanism}

\begin{figure}[h]
  \centering
  \includegraphics[width=0.49\linewidth]{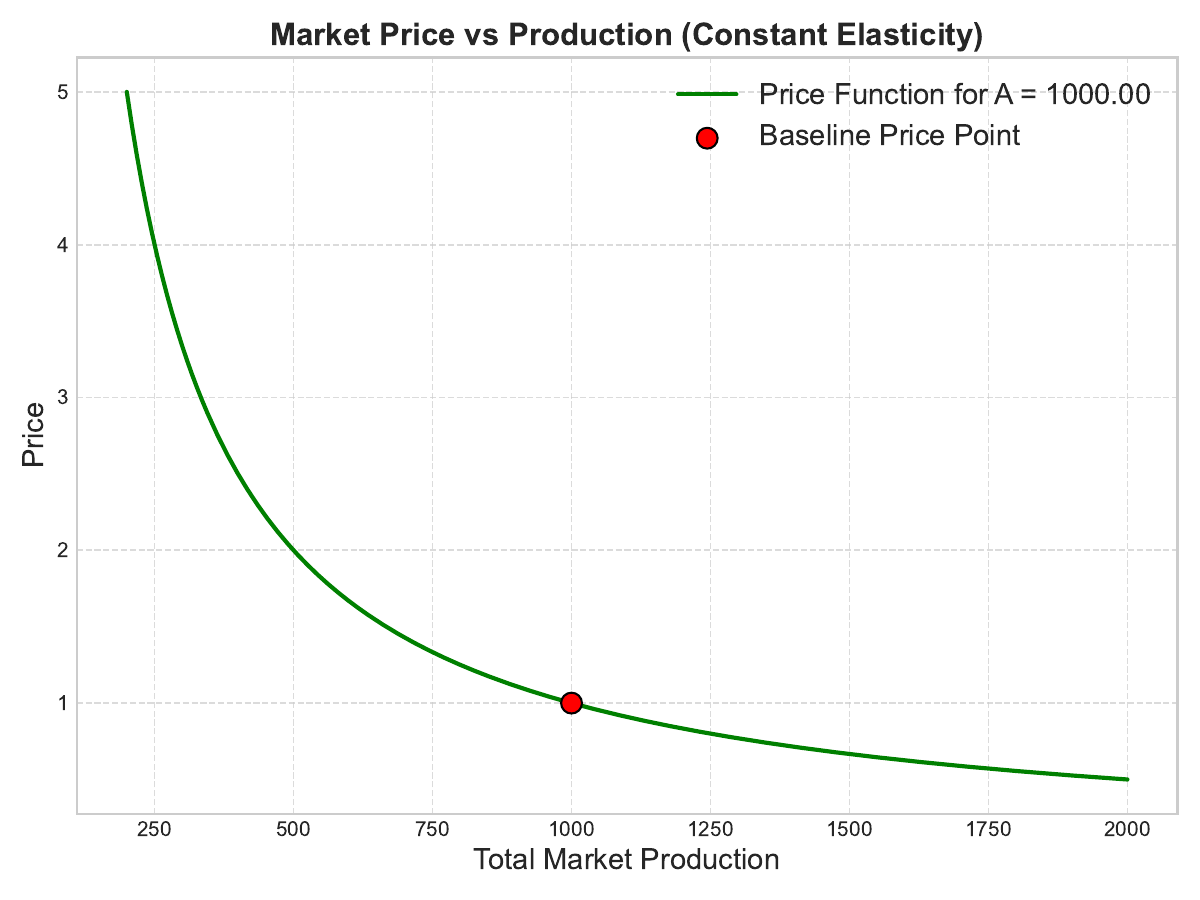}
  \includegraphics[width=0.49\linewidth]{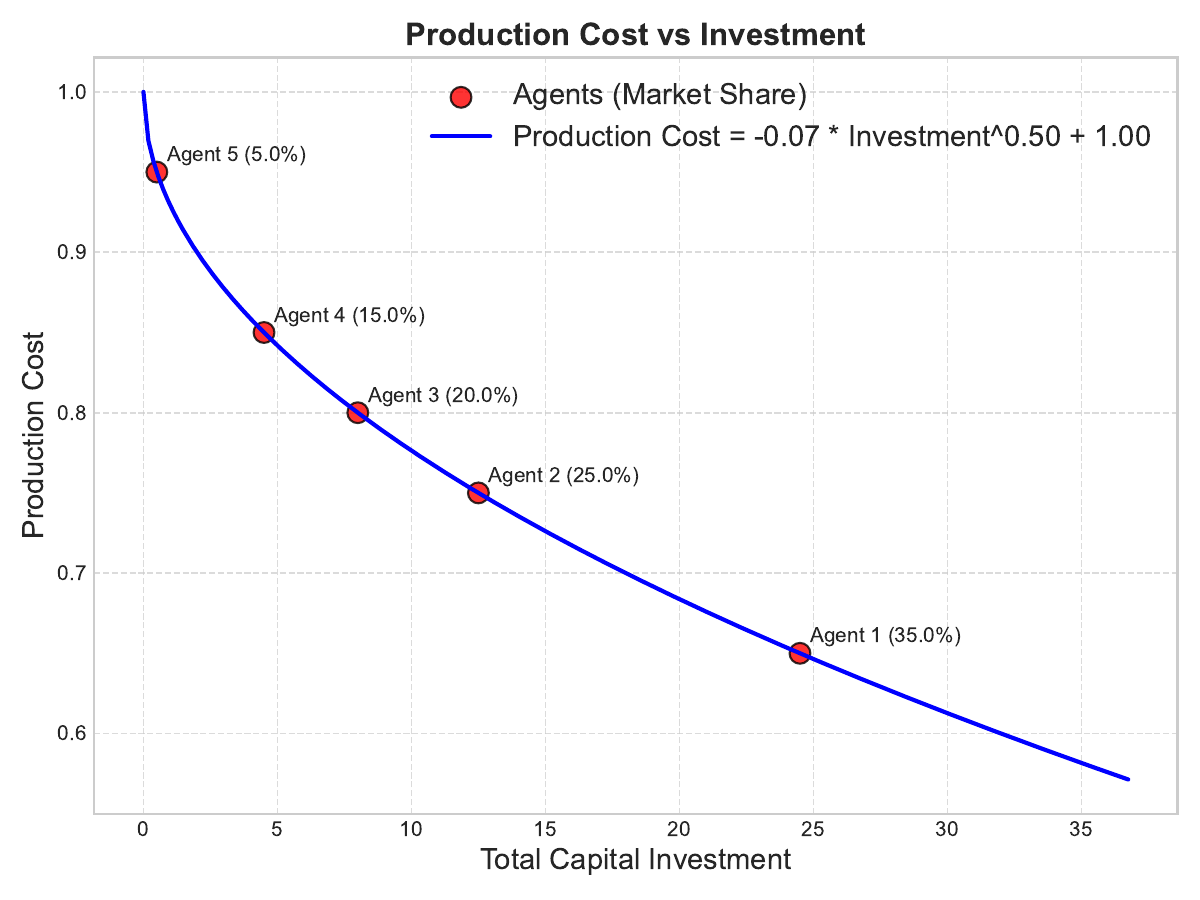}
  \caption{The market feedback mechanism: How market prices and firm-level production-costs are determined, given firms' production $(q_i)$ and investment $(b_i)$ decisions. The red points show the Nash positions.}
  \label{fig:market_mech}
\end{figure}

\Cref{fig:market_mech} shows the market feedback mechanism, specific to the five-firm market considered in \Cref{sec:oligopoly} (See \Cref{tab:agent_info_5} in \Cref{app:params} for associated parameter details).
The plot on the left illustrates how prices change with production (Eq. \eqref{eq: price}), while the plot on the right depicts production-costs as a function of absolute capital investment (Eq. \eqref{eq: cobb-douglas}). The points represent firms' positions (market shares in brackets) when they invest the maximum allowed amount—20\%  of their baseline (naive) profits. If a firm reduces this investment percentage (i.e., a point shifts left on the graph), the graph indicates the corresponding increase in production-costs.

\subsubsection{Illustrative runs for various market cases}
\Cref{fig:5player_llms} illustrates market dynamics over 300 periods when all five agents are LLM-based. \Cref{fig:5player_one_br} and \Cref{fig:5player_br} depict similar scenarios, with the former involving four LLM agents interacting with one BR agent and the latter featuring three LLM agents alongside two BR agents.
\begin{figure}[H]
  \centering
  \includegraphics[width=0.78\linewidth]{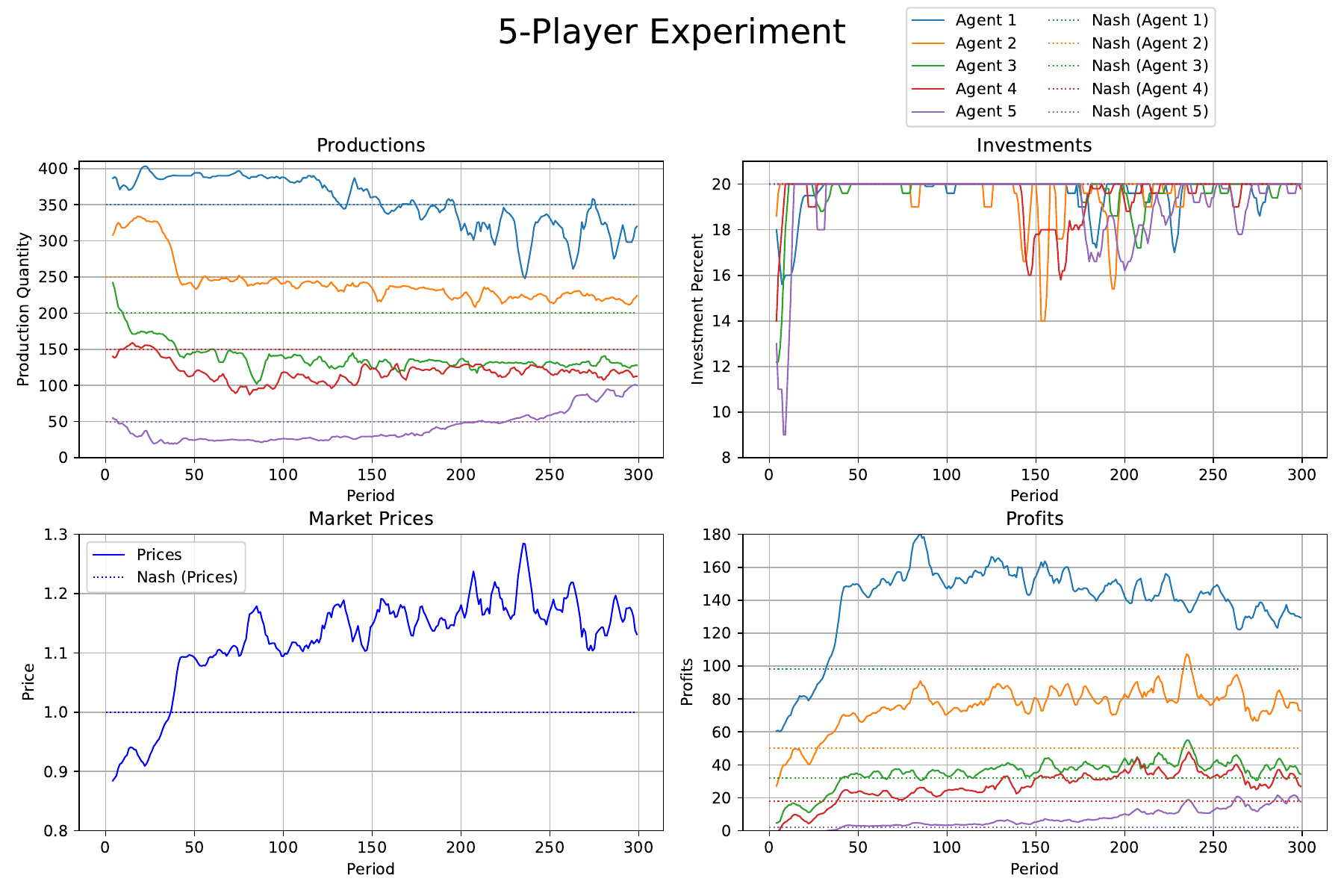}
  \caption{\textbf{LLMs vs LLMs: } Decisions and Market Dynamic in a 5-firm experiment.}
  \label{fig:5player_llms}
\end{figure}
\begin{figure}[h]
  \centering
  \includegraphics[width=0.72\linewidth]{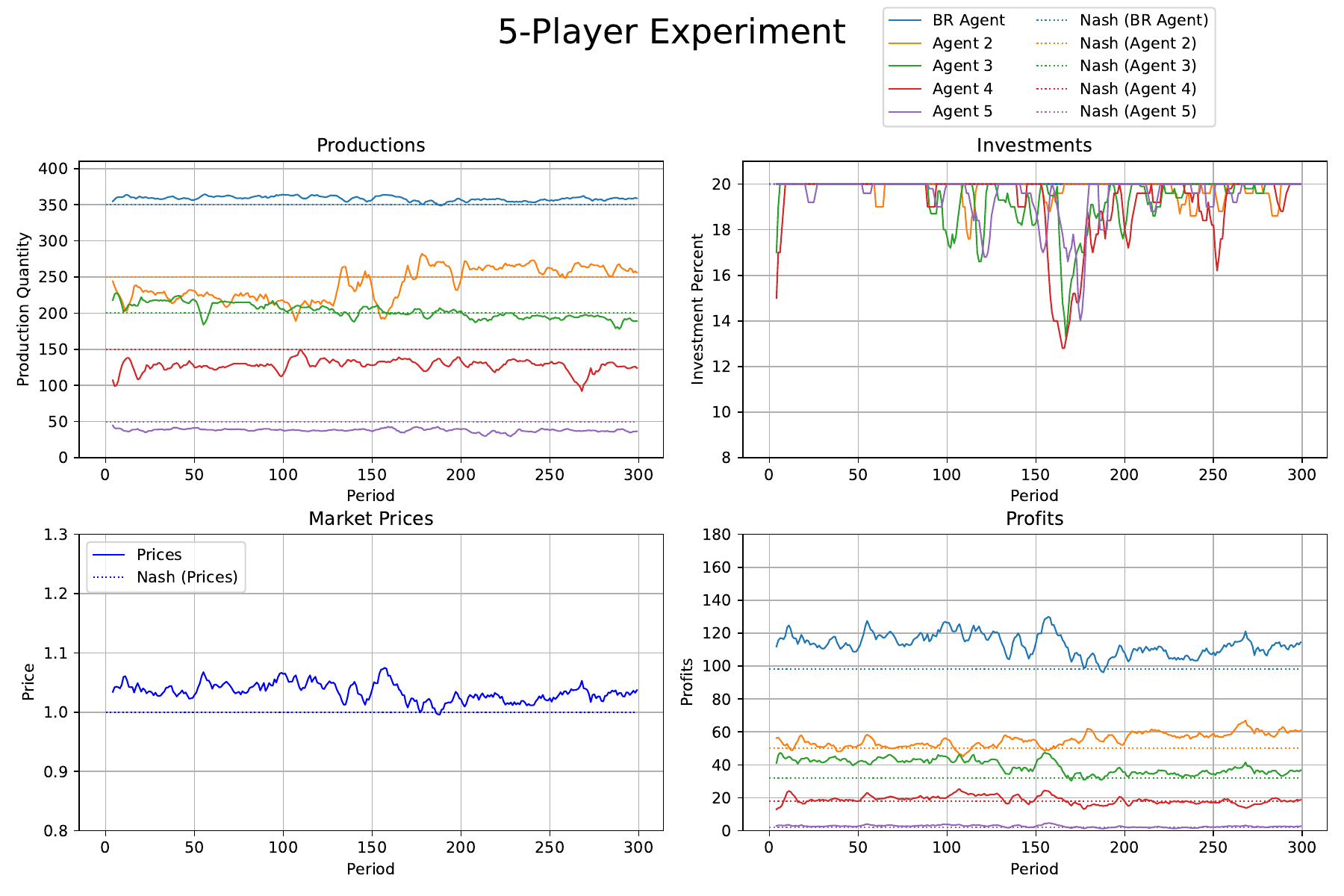}
  \caption{\textbf{4 LLM vs 1 Best-Response agents: } Decisions and Market Dynamic in a 5-firm experiment.}
  \label{fig:5player_one_br}
\end{figure}
\begin{figure}[H]
  \centering
\includegraphics[width=0.72\linewidth]{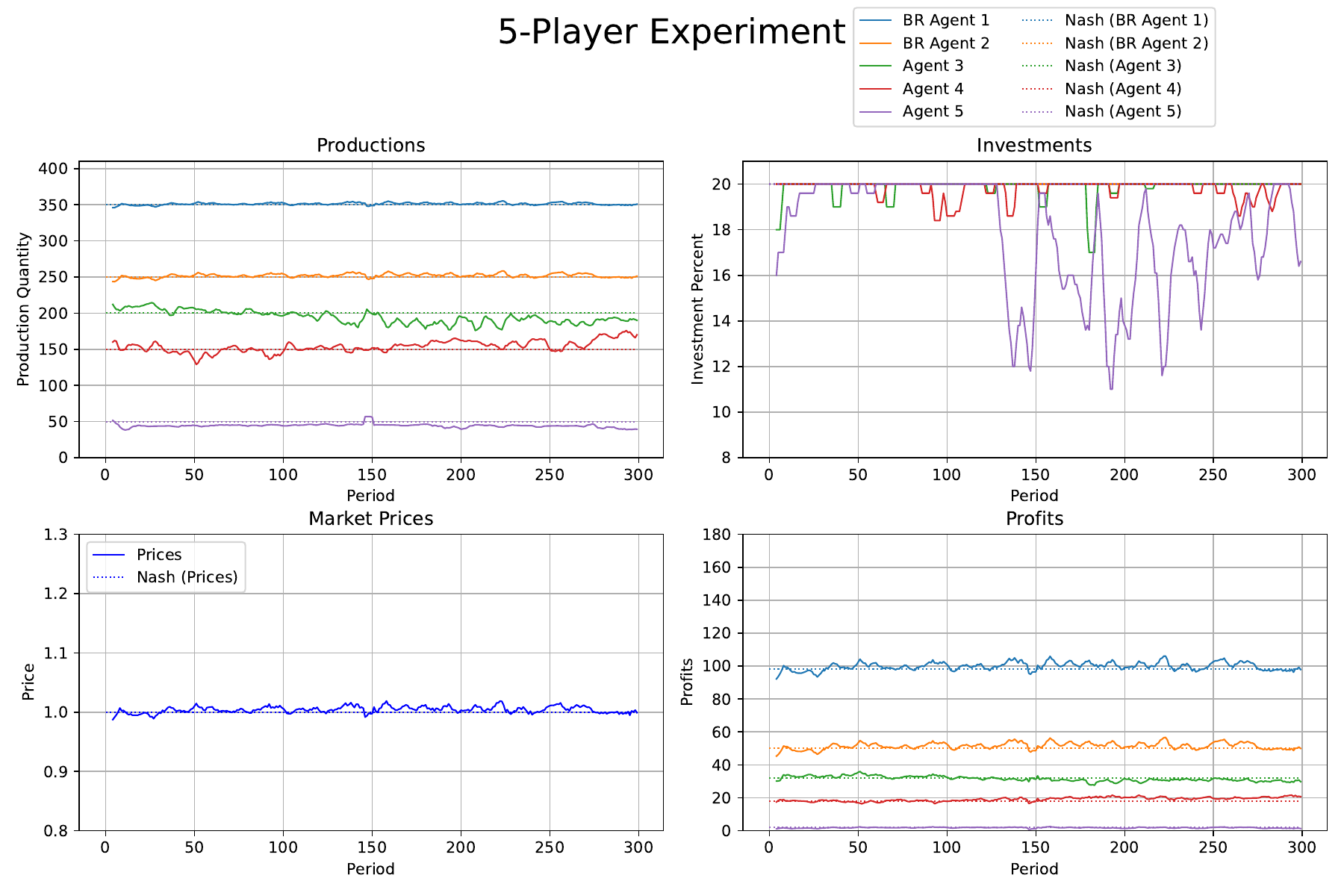}
  \caption{\textbf{3 LLM vs 2 Best-Response agents: } Decisions and Market Dynamic in a 5-firm experiment.}
  \label{fig:5player_br}
\end{figure}

Consistent with our two-firm experiments, \Cref{fig:5player_br} shows prices converging toward Nash equilibrium. In contrast, \Cref{fig:5player_llms} exhibits supra-competitive pricing. \Cref{fig:5player_one_br} presents an intermediate case—prices remain above Nash equilibrium but lower than in the fully unregulated market. The aggregated results from multiple such runs are discussed in \Cref{sec:oligopoly}.

\subsection{Additional Results: 6-Firm Heterogeneous Market}
\label{app:6player_heterogeneous}

In addition to our 5-firm experiments in \Cref{sec:oligopoly}, we analyze a slightly bigger heterogeneous market of 6 firms. The 5-firm market consisted of firms with the following market shares: 35\%, 25\%, 20\%, 15\%, and 5\%. 
For the 6-firm market, we consider an oligopolistic setting with the following market shares: 31.4\%, 20.2\%, 16.9\%, 14.5\%, 12.5\%, and 4.6\%. These values intentionally include floating-point percentages to test how the LLM agents respond under more granular market conditions.

We run three experiments in the unregulated setting and two experiments each for cases where (i) only the top firm is regulated and (ii) the top two firms are regulated. Each experiment consists of 300 periods, and we focus on insights from the last 50 periods.

Overall, the results reinforce our key observations from \Cref{sec:oligopoly}: 
\begin{itemize}
    \item LLM agents exhibit tacit collusion in the unregulated market, with prices up to a run-average of 15\% higher than Nash.
    \item When only the top firm is regulated to play a best-response (BR) strategy, the average prices are contained within 5\% above Nash.
    \item When the top two firms are regulated, market prices converge to Nash equilibrium levels.
\end{itemize}
The detailed outcomes for this 6-firm market are presented below in \Cref{app:6player_all_llms} and \Cref{app:6player_regulation}. A sample 6-firm run is reproduced in \Cref{app:6_player_run}.

\subsubsection{All LLMs}
\label{app:6player_all_llms}

\begin{figure}[ht!]
    \centering
    \includegraphics[width=0.49\linewidth]{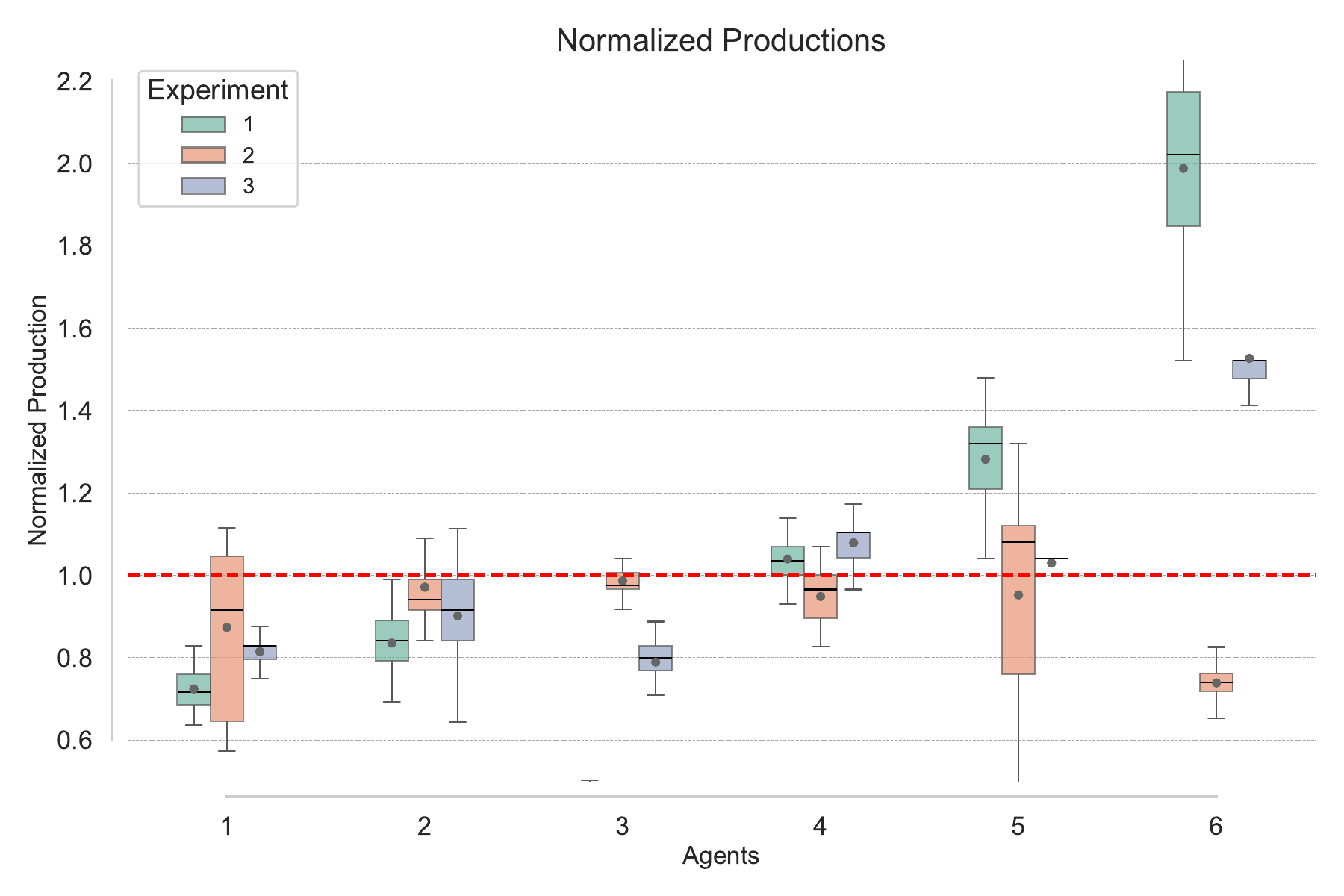}
    \includegraphics[width=0.49\linewidth]{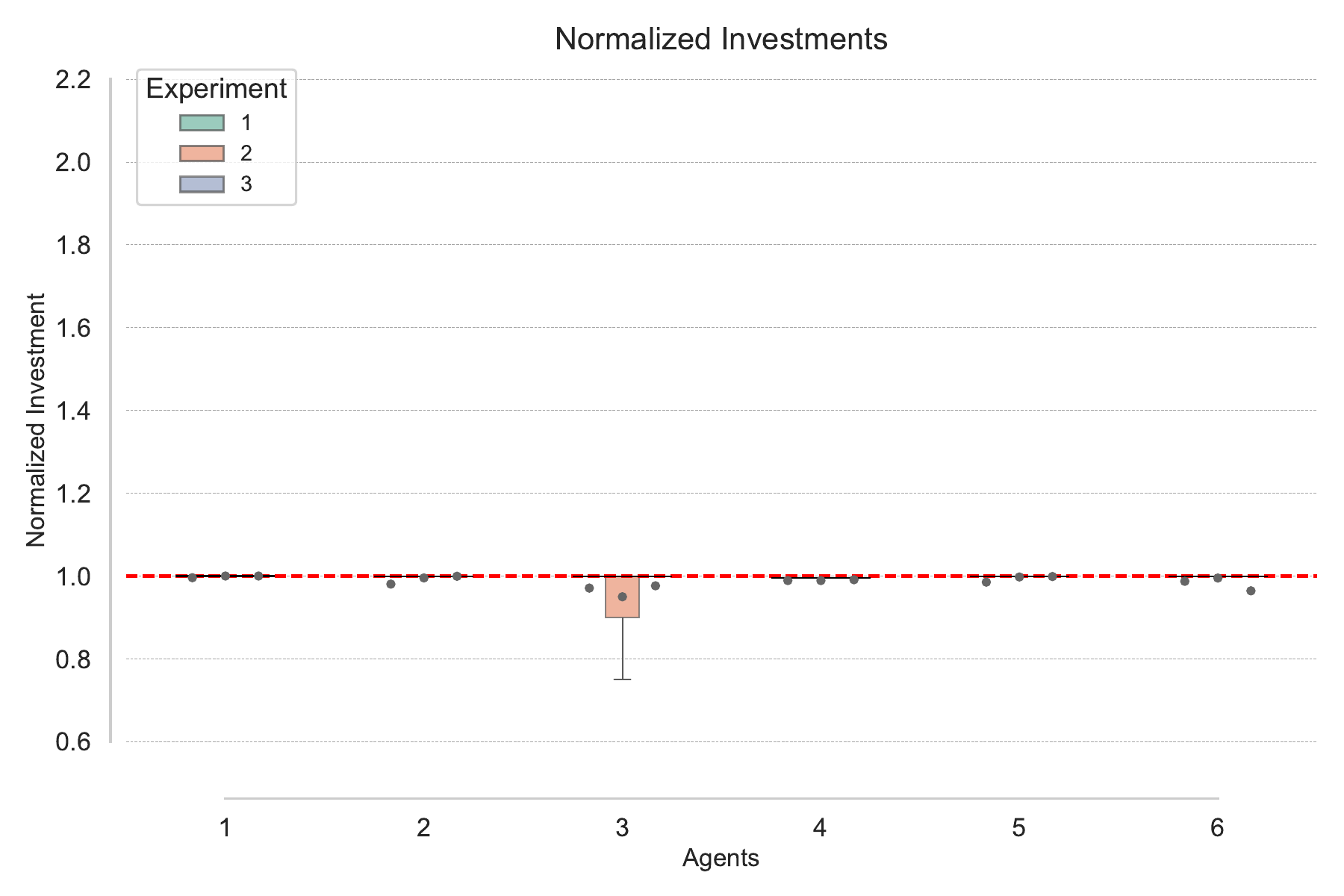}
    \caption{\textbf{LLMs vs LLMs actions:} Agents 1 to 6 have decreasing market shares.}
    \label{fig:LLMvsLLM_6_agents_1}
\end{figure}

\begin{figure}[ht!]
    \centering
    \includegraphics[width=0.40\linewidth]{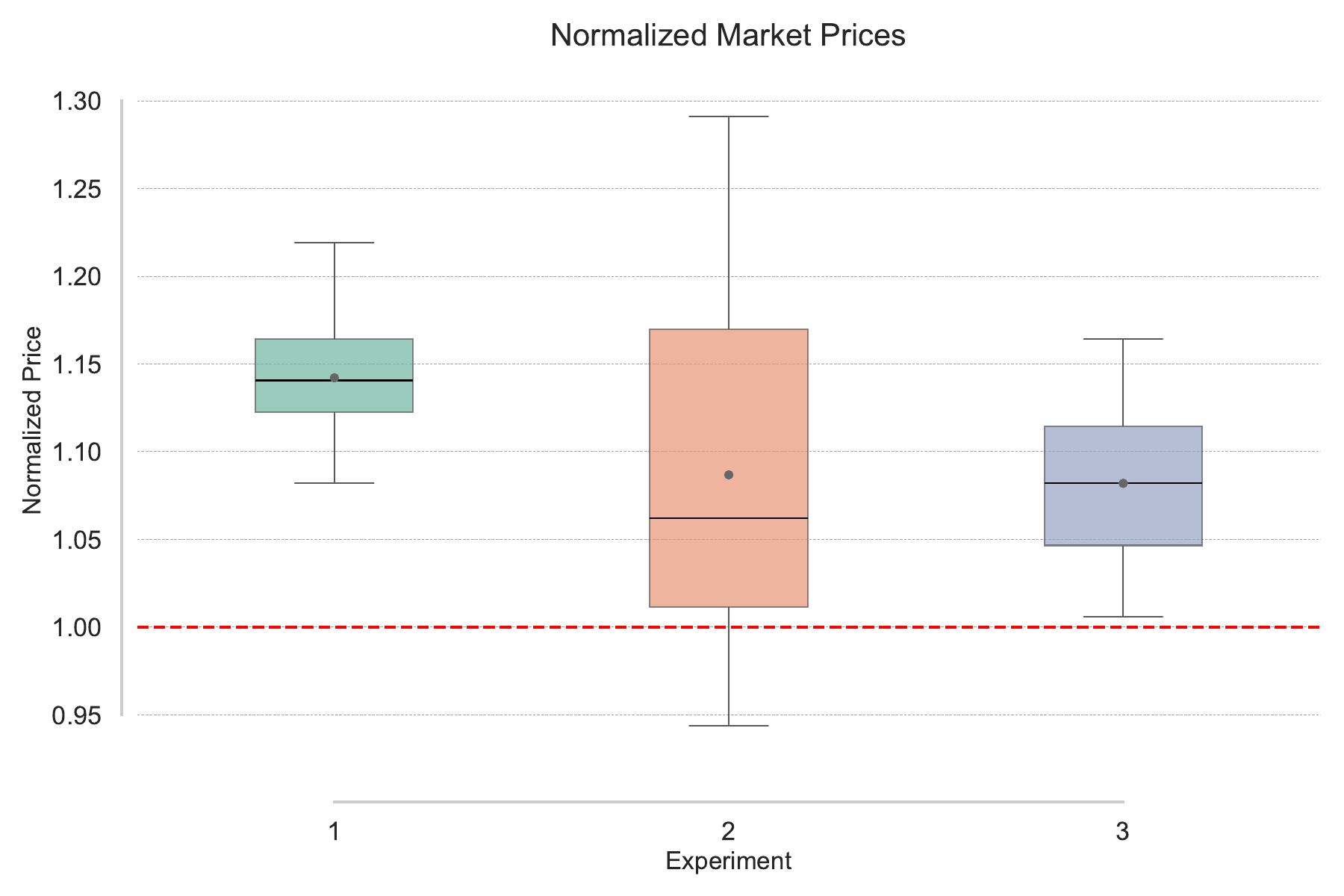}
    \includegraphics[width=0.59\linewidth]{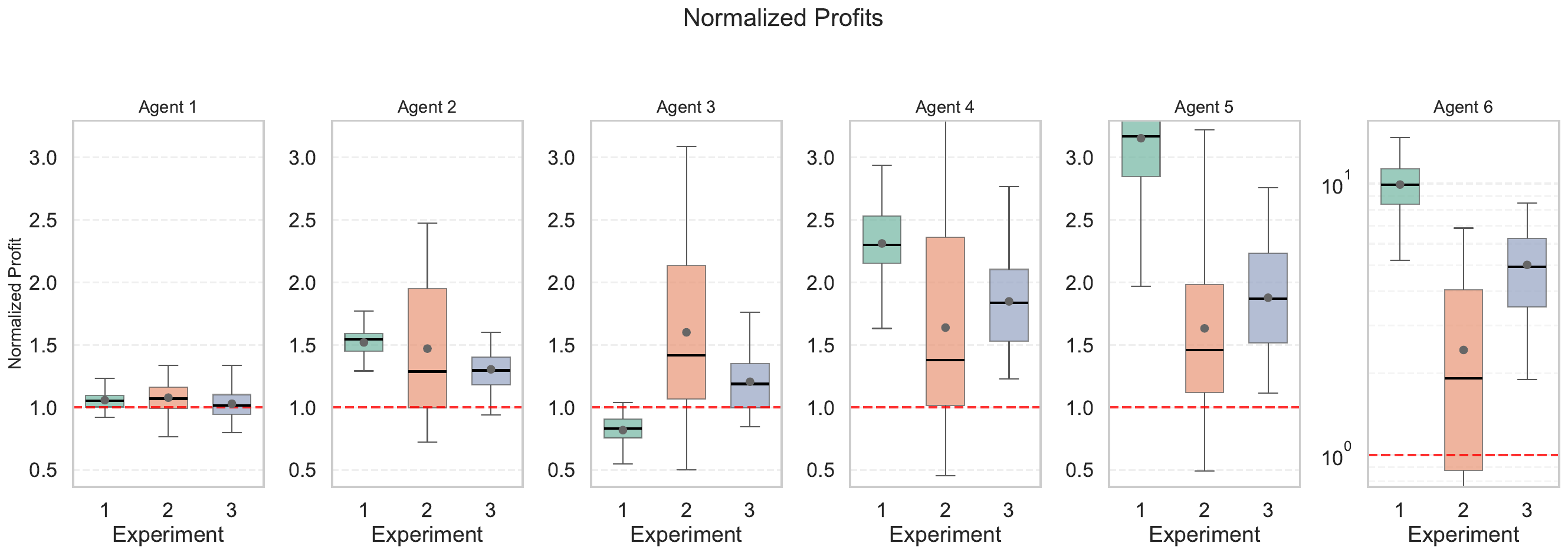}
    \caption{\textbf{LLMs vs LLMs Market dynamics }}
    \label{fig:LLMvsLLM_6_agents_2}
\end{figure}

\noindent
Despite differences in market shares, we observe sustained supra-competitive pricing in the 6-firm market as well. \Cref{fig:LLMvsLLM_6_agents_1,fig:LLMvsLLM_6_agents_2} show the decisions made by LLMs and their market impact in unregulated markets. In \Cref{fig:LLMvsLLM_6_agents_1}, see that the top 3 firms almost always produce below Nash. The smaller firms can scale better, although their overall production remains comparatively low. All firms plan investments on the higher side (about 90—100\% of the ideal values), as it remains optimal to invest even if they plan on low production. 
\Cref{fig:LLMvsLLM_6_agents_2} shows that all prices remain supra-competitive and increase from about 8\% to 20\% compared to Nash. All firms earn either Nash-level or above-Nash profits.

\subsubsection{Regulated Market}
\label{app:6player_regulation}

\begin{figure}[ht!]
    \centering
    \includegraphics[width=0.49\linewidth]{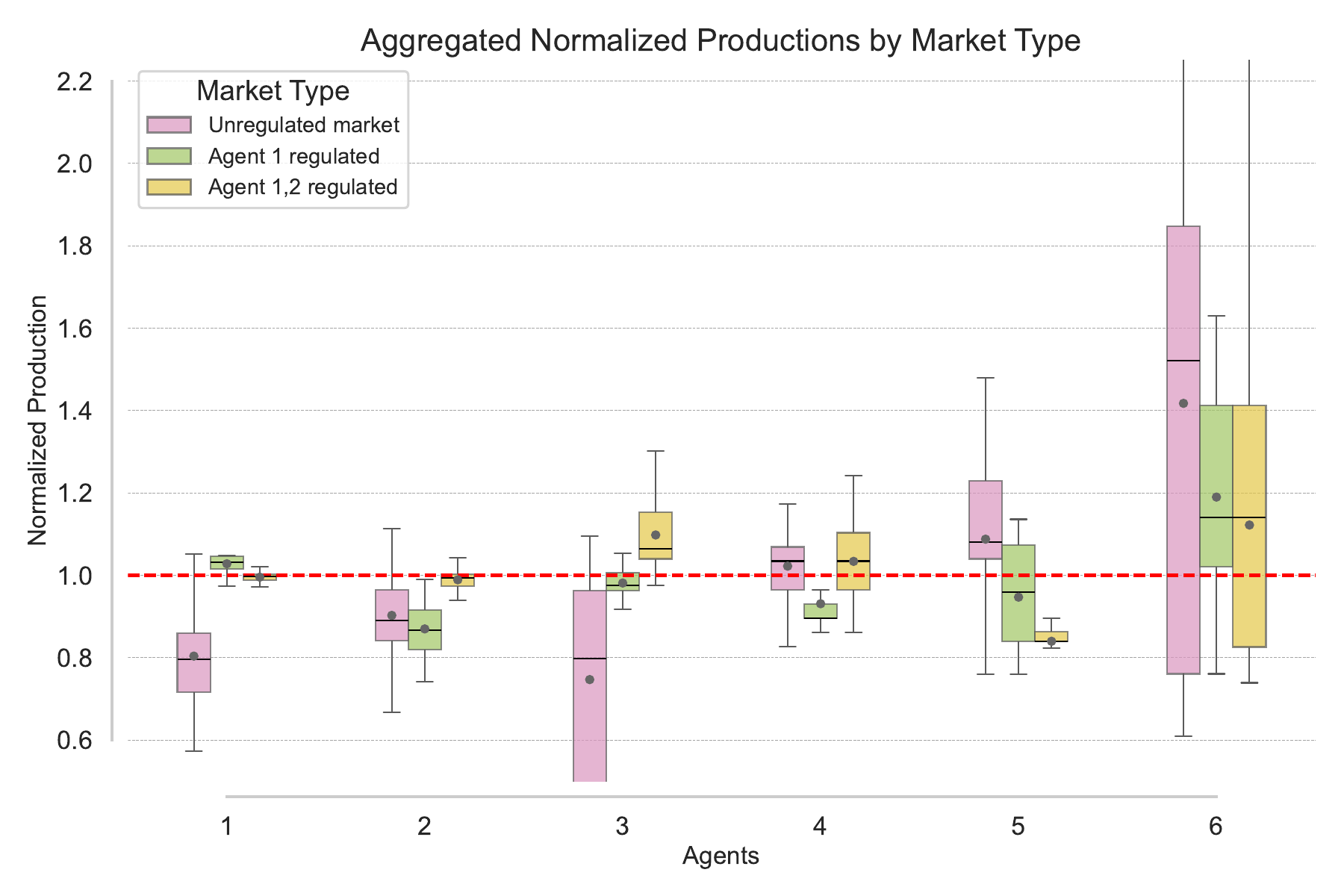}
    \includegraphics[width=0.49\linewidth]{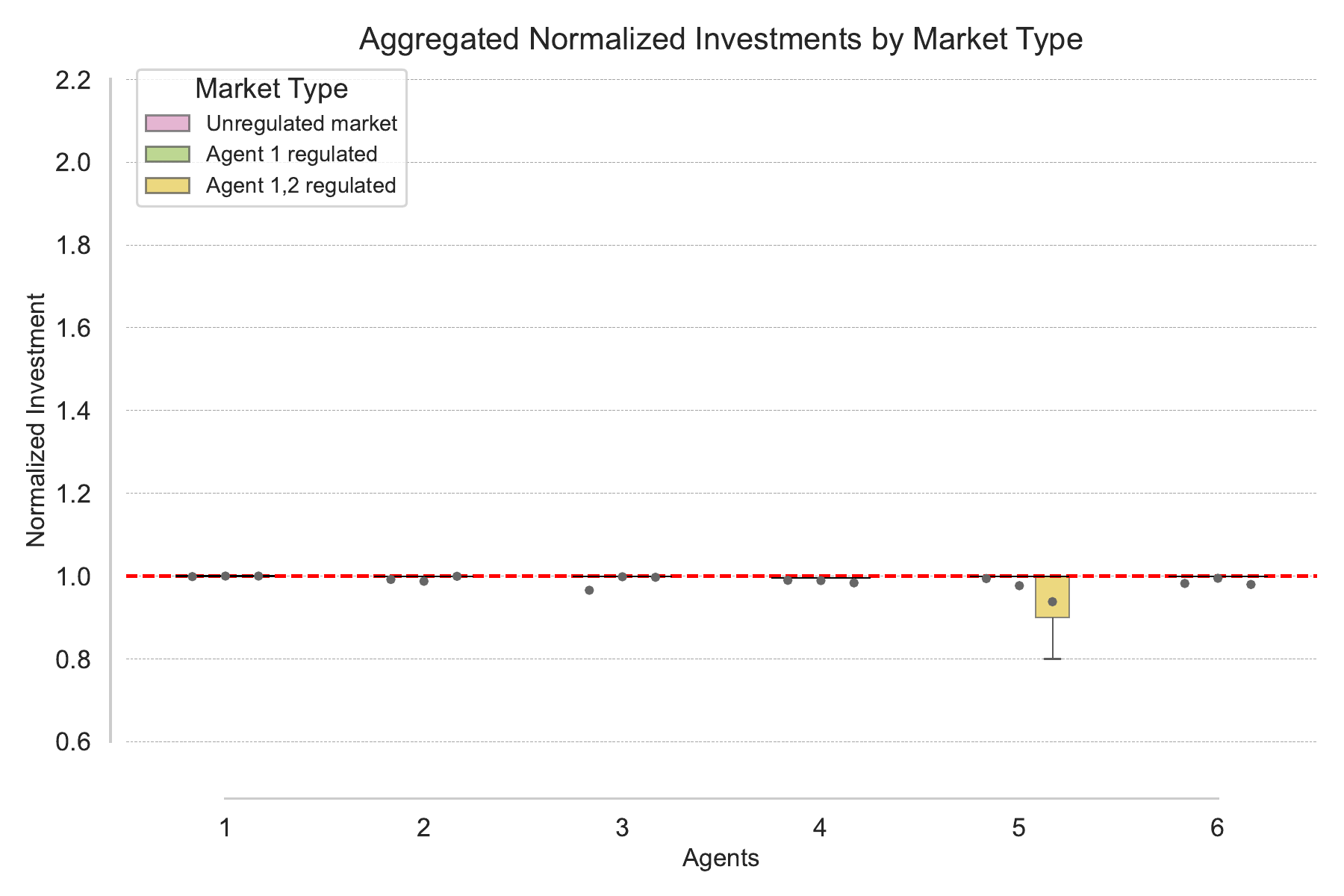}
    \caption{\textbf{Actions in unregulated and regulated markets:} The first group of three experiments (left group) is unregulated. The middle and the right groups show 2 regulated experiments each, where the top firm and the top 2 firms, respectively, are required to play BR to the market.}
    \label{fig:6_players}
\end{figure}

\begin{figure}[ht!]
    \centering
    \includegraphics[width=0.40\linewidth]{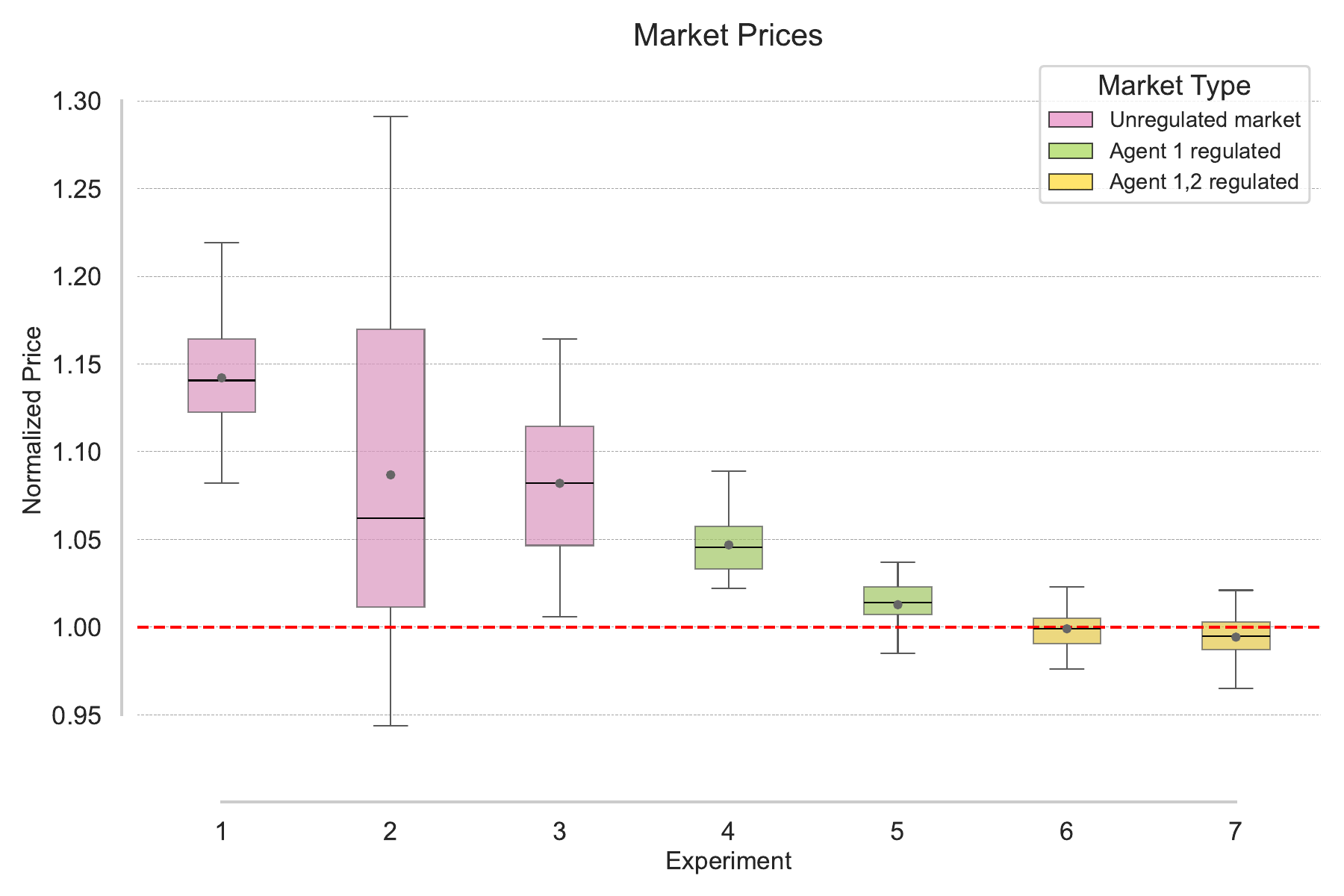}
    \includegraphics[width=0.59\linewidth]{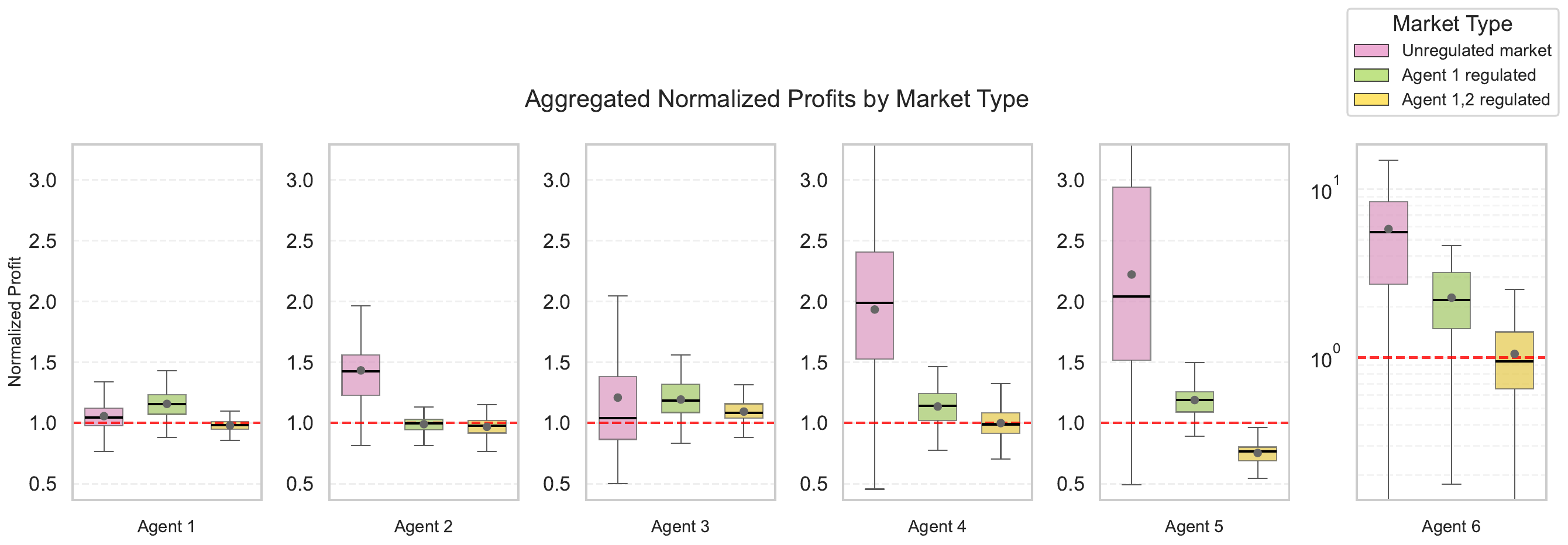}
    \caption{\textbf{Market Dynamics in unregulated and regulated markets} }
    \label{fig:6_player_regulation}
\end{figure}

\noindent
As shown in \Cref{fig:6_players,fig:6_player_regulation}, regulating the top firms to follow best-response strategies leads to pricing outcomes close to the Nash equilibrium. Overall, production decisions under regulation are closer to Nash solutions than in unregulated settings. Investments also remain close to optimal, especially for larger firms.  This reduces the above-Nash profits observed in the unregulated market and helps contain supra-competitive pricing. In conclusion, the findings from the six-firm experiment align with those of the five-firm market, leading to similar conclusions.

\subsubsection{Illustrative Run: 6-Firm Market} \label{app:6_player_run}

\begin{figure}[H]
  \centering
  \includegraphics[width=0.7\linewidth]{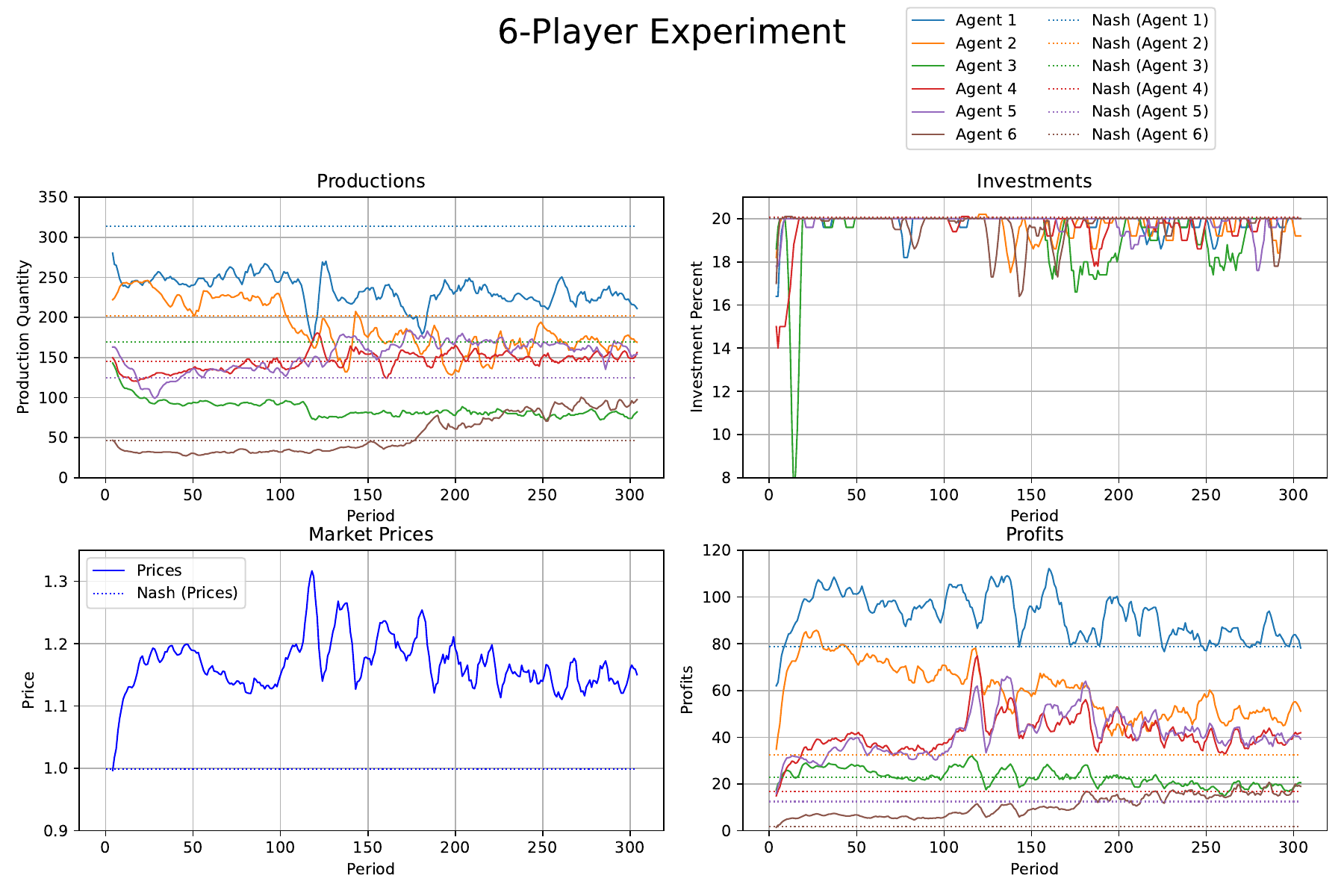}
  \caption{\textbf{LLMs vs LLMs:} Decisions and Market Dynamic in a 6-Firm experiment.}
  \label{fig:6player_decisions}
\end{figure}
\begin{figure}[H]
  \centering
  \includegraphics[width=0.7\linewidth]{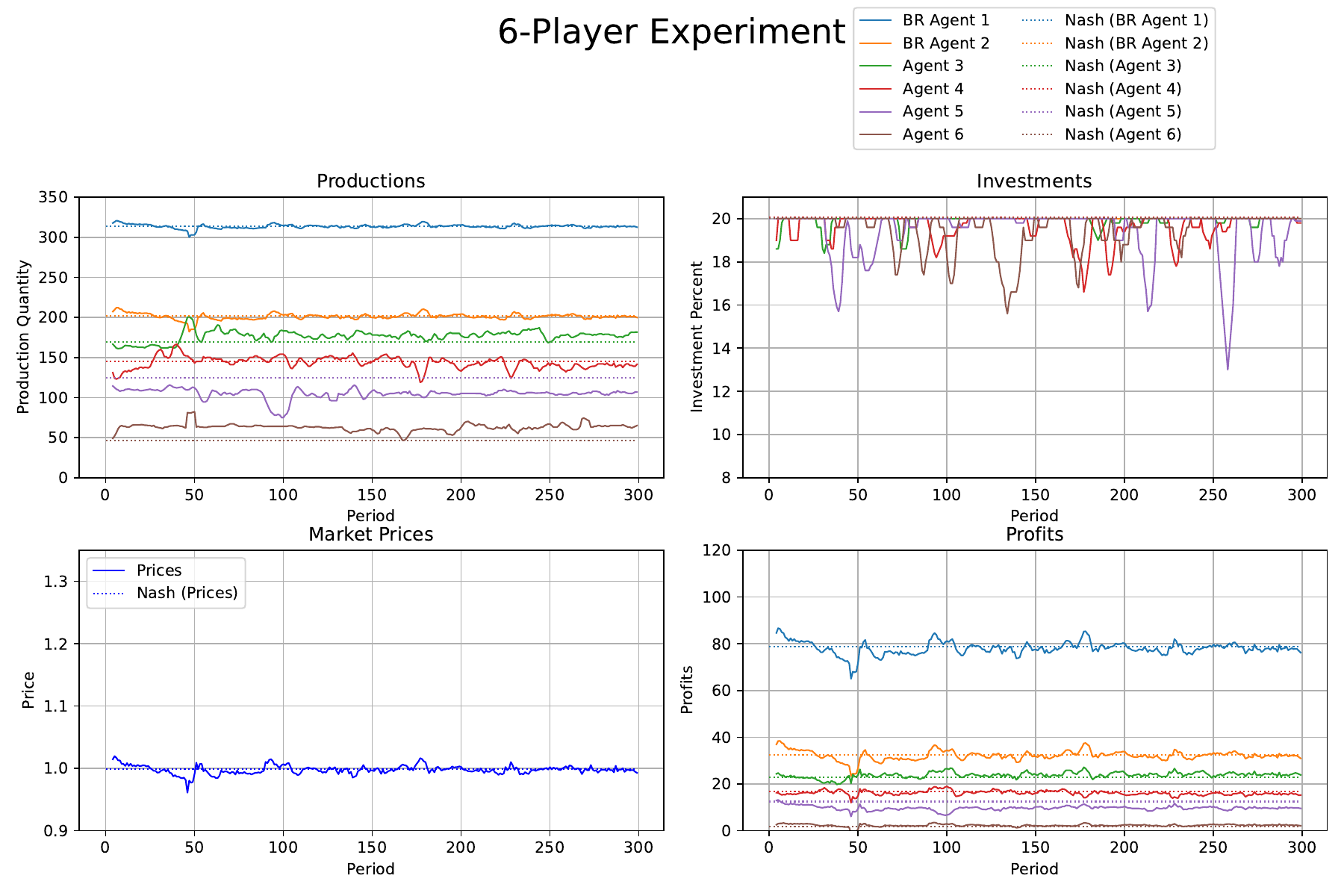}
  \caption{\textbf{4 LLMs vs 2 BR (best response) agents:} Decisions and Market Dynamic in a 6-Firm experiment.}
  \label{fig:6_player_market}
\end{figure}

\end{document}